\documentclass[a4paper,onecolumn,11pt]{quantumarticle}
\pdfoutput=1

\usepackage{lipsum}
\usepackage{amsfonts}
\usepackage{graphicx}
\usepackage{epstopdf}
\usepackage{algorithmic}
\usepackage{multibib}
 \usepackage[sort,nocompress]{cite}
\ifpdf
  \DeclareGraphicsExtensions{.eps,.pdf,.png,.jpg}
\else
  \DeclareGraphicsExtensions{.eps}
\fi

\usepackage[T1]{fontenc}
\usepackage[utf8]{inputenc}
\usepackage{amsmath}
\usepackage{amsthm}
\usepackage{amssymb}
\usepackage{mathtools}
\usepackage{multirow}
\usepackage{braket}
\usepackage{url}
\usepackage{mathdots}
\usepackage[export]{adjustbox}
\usepackage{booktabs}
\usepackage{siunitx}
\usepackage{caption}
\usepackage{hyperref}
\usepackage{cleveref}
\usepackage{color} 
\usepackage[vlined, ruled, boxed]{algorithm2e}
\usepackage{bm}
\usepackage{diagbox}
\usepackage{changepage}
\usepackage{soul}
\usepackage{mathdots}
\usepackage{subfig}

\usepackage{makecell}

\newcommand{\llbracket}{{[\![}}
\newcommand{\rrbracket}{{]\!]}}

\makeatletter

\DeclareUnicodeCharacter{00A0}{ }

\SetKwProg{Fn}{Function}{}{}

\makeatother

\usepackage[english]{babel}
\usepackage[normalem]{ulem}

\newtheorem{mydef}{Definition}
\newtheorem{property}{Property}
\newtheorem{theorem}{Theorem}
\newtheorem{corollary}{Corollary}

\author{Timothée Goubault de Brugière}
\affiliation{Université de Lorraine, CNRS, Inria, LORIA, F-54000 Nancy, France}
\author{Simon Martiel}
\affiliation{Atos Quantum Lab, Les Clayes-sous-Bois, France}
\author{Christophe Vuillot}
\affiliation{Université de Lorraine, CNRS, Inria, LORIA, F-54000 Nancy, France}

\title{A graph-state based synthesis framework for Clifford isometries}

\begin{document}

\newcommand\redsout{\bgroup\markoverwith{\textcolor{red}{\rule[0.5ex]{2pt}{1pt}}}\ULon}

\maketitle

\begin{abstract}

We tackle the problem of Clifford isometry compilation, i.e, how to synthesize a Clifford isometry into an executable quantum circuit. We propose a simple framework for synthesis that only exploits the elementary properties of the Clifford group and one equation of the symplectic group. We highlight the versatility of our framework by showing that several normal forms of the literature are natural corollaries. We recover the state of the art two-qubit gate depth necessary for the execution of a Clifford circuit on an LNN architecture, concomitantly with another work. 
We also propose practical synthesis algorithms for Clifford isometries with a focus on Clifford operators, graph states and codiagonalization of Pauli rotations. Benchmarks show that in all three cases we improve the 2-qubit gate count and depth of random instances compared to the state-of-the-art methods. We also improve the execution of practical quantum chemistry experiments.

\end{abstract}

\section{Introduction}

Clifford operators represent one of the most useful and one of the most studied subclass of quantum operators. Clifford operators are involved in numerous applications such as quantum error correction \cite{gottesman1997stabilizer}, randomized benchmarking protocols \cite{knill2008randomized,magesan2011scalable}, quantum state distillation \cite{bravyi2005universal, knill2005quantum}, the study of entanglement \cite{bennett1996mixed}. It is known that Clifford operators can be generated by the Hadamard gates, the Phase gates and the CNOT gate \cite{aaronson2004improved}. Moreover the simulation of such circuits, the so-called Clifford circuits, can be done efficiently on a classical computer \cite{gottesman1998heisenberg} and there is a one-to-one mapping between the Clifford operators and the $2n \times 2n$ binary matrices from the symplectic group $Sp(2n, \mathbb{F}_2)$. 

Our work lies in the domain of quantum compilation for the quantum circuit model, where one wants to generate quantum circuits for a given hardware to execute a given quantum algorithm. Given the variety of quantum hardwares, each having its own set of universal gates, its own costly resources, and given the different approaches for quantum computation, namely a fault-tolerant approach \cite{shor1996fault,preskill1998fault} or a NISQ approach \cite{NISQC}, today's compilers have to be able to
generate and optimize quantum circuits for several metrics (number of non Clifford gates, number of entangling gates, depth), several set of gates (Clifford+T, MS+SU(2), etc.) and several architectures (Rigetti \cite{manenti2021full}, IBM \cite{chowImplementingStrandScalable2014,chamberlandTopologicalSubsystemCodes2020}, Google \cite{arute2019quantum}, etc.). While not universal for quantum computation, the Clifford group plays a significant role in quantum compilation for several types of hardware \cite{van2021constructing,maslov2018use,martiel2020architecture} in both a
NISQ of a fault-tolerant setting and notably for the hardwares using the Clifford+T gate set, one of the most studied gate set for compilation problems. In this setting, the available multi-qubit gates, generally the 2-qubit CNOT gate or the CZ gate, belong to the Clifford group and looking at a quantum circuit as a succession of Clifford operators separated by layers of T-gates is one way to express the problem of optimizing the 2-qubit gate cost as a Clifford circuit optimization problem. Given that the abstract and compact representation of Clifford operators makes it easier to design efficient synthesis algorithms the Clifford operator synthesis problem emerges naturally as an essential brick for quantum circuit optimization, and notably to reduce the 2-qubit gate cost.

The structure of the Clifford group has been deeply studied in the literature. Notably, many normal forms for Clifford circuits have been proposed \cite{aaronson2004improved, maslov2018shorter, duncan2020graph, bravyi2021hadamard, bataille2021reduced}. Those normal forms express any Clifford operator as a series of stages containing one particular gate. For instance, the first normal form proposed by Aaronson and Gottesman is -H-CX-P-CX-P-CX-H-P-CX-P where -H- stands for a stage of Hadamard gates, -P- a stage of phase gates, and -CX- a stage of CNOT gates \cite{aaronson2004improved}. An important part of the work in Clifford circuits synthesis during the last decade was to propose shorter normal forms, resulting in shorter or shallower circuits with respect to the 2-qubit gates. Normal forms for the preparation of stabilizer states (states that can be prepared with a Clifford circuit) have also been proposed \cite{aaronson2004improved, van2010classical, garcia2014geometry}. Those normal forms, while being universal, impose a strict structure on the produced circuits and may fail to find more efficient implementations for a given operator. Notably the most recent normal forms expose the role of a particular class of quantum operators, the so-called phase polynomials, in the implementation of a Clifford operator \cite{maslov2018shorter, bravyi2021hadamard}. Phase polynomials are now extensively studied in the literature and the best synthesis algorithms show that a stage by stage implementation of phase polynomials is irrelevant \cite{amy2018controlled, vandaele2021phase}. We can expect that a more efficient approach for Clifford synthesis that does not rely on a stage by stage implementation can also be found. 

Optimal synthesis of Clifford operators through brute-force search has been recently proposed \cite{bravyi20206,DBLP:conf/aspdac/SchneiderBW23,DBLP:conf/qce/PehamBKWB23}, but for the moment the method is limited to small Clifford circuits, up to 26 qubits. Other works in the literature also focus on directly optimizing Clifford circuits through template matchings or peephole optimizations \cite{kliuchnikov2013optimization, bravyi2021clifford}. These methods cannot handle the synthesis from an abstract representation nor leverage the global structure of these operators with potential efficiency and performance loss.

Overall, it lacks an efficient approach to synthesize Clifford operators other than by using normal forms or the unpractical tableau formalism. Furthermore, there are many cases in the literature where only a partial implementation of a Clifford operator is necessary\textcolor{red}{.} \textcolor{red}{O}bviously the synthesis of stabilizer states is one example, but more generally the synthesis of Clifford isometries where some of the input qubits are in the state $\ket{0}$ is of particular interest in some error-correction protocols such as magic state distillation \cite{bravyiUniversalQuantumComputation2005, bravyiMagicstateDistillationLow2012}. Again, to our knowledge, no work in the literature proposes to generalize the concept of normal forms and synthesis frameworks to Clifford isometries.

The main result of our paper is a generalized synthesis framework for Clifford isometries. Our framework is strongly inspired by the results obtained via the ZX-calculus for the optimization of quantum circuits \cite{duncan2020graph} although the derivation of our framework is done without any use of the ZX-calculus. Indeed, our proof uses only the properties of the Clifford group and the symplectic matrices used to represent Clifford operators. Our framework shares also some similarities with
the normal form proposed in \cite{bravyi2021hadamard}. Where our framework stands out is first of all in its simplicity: it only requires to reduce a symmetric boolean matrix into a suitable defined identity matrix with the use of a few elementary operations, each corresponding to one elementary Clifford gate. Notably it naturally includes the Hadamard gate in the available operations, which is the gate that is missing when considering a Clifford operator with stages of phase polynomials. We
also highlight its versatility by showing that many normal forms proposed in the literature can be trivially recovered. Similarly to \cite{bravyi2021hadamard}, we also characterize more precisely the operators involved in the synthesis of one operator when the gate set is restricted to $\{CNOT, CZ, S\}$. 
We show that any Clifford operation on $n$ qubits can be executed in two-qubit gate depth $7n-2$ on a Linear Nearest Neighbour (LNN) architecture. Our result is concomitant with the $7n+2$ depth (then improved to $7n-4$) of another article in the literature \cite{maslov2023cnot}. Overall, both work improve the previous best result of $9n$ \cite{duncan2020graph, bravyi2021hadamard}.

Finally we apply our framework to the design of practical synthesis algorithms for Clifford isometries. We propose two versions: one that optimizes the 2-qubit gate count and one that optimizes the 2-qubit gate depth. In our benchmarks we will focus on two special cases: the synthesis of graph states and the synthesis of Clifford operators. On random instances, we show that we outperform the state of the art. We also extend our framework to the codiagonalization of groups of commuting Pauli rotations, we apply our algorithms to a set of benchmarks from quantum chemistry and report improvements over the state of the art. 

\section{Background}

\subsection{Quantum isometries}

An isometry is a generalization of quantum states and quantum operators. They map Hilbert spaces of possibly two different dimensions and they preserve the inner products, similarly to unitary matrices.

\begin{mydef}
Let $0 \leq k \leq n$ be two integers. A $k$ to $n$ isometry can be represented by a $2^n \times 2^k$ complex matrix $V$ such that 
\[ V^{\dag}V = I_{2^k \times 2^k}. \]
\end{mydef}

In other words, isometries can be seen as the action of a quantum operator on a quantum memory where some of the qubits are in the state $\ket{0}$ --- those qubits can be regarded as ancillary qubits. Generally the compilation of isometries result in shorter circuits than systematically synthesizing the full quantum operator because the ancillary qubits in a fixed input state offer some degrees of freedom in the synthesis. Clearly, a state preparation is usually less costly than a full operator compilation.

In this paper we focus on the synthesis of a particular subclass of isometries: Clifford isometries, i.e, isometries that can be generated by a Clifford operator. 

\subsection{Pauli matrices and the Clifford group}

The Pauli matrices are defined by
\[ I=\begin{bmatrix}1&0\\0&1\end{bmatrix}, X = \begin{bmatrix} 0 & 1 \\ 1 & 0 \end{bmatrix}, Z = \begin{bmatrix} 1 & 0 \\ 0 & -1 \end{bmatrix}, Y = \begin{bmatrix} 0 & -i \\ i & 0 \end{bmatrix} \]
and they obey the following equations: 
\[ X^2 = Y^2 = Z^2 = I = \begin{bmatrix} 1 & 0 \\ 0 & 1 \end{bmatrix}, \]
\[ XY = iZ, \]
\[ ZX = iY, \]
\[ YZ = iX. \]

Therefore, the Pauli matrices generate a group where any member can be written as
\[ i^k \{I, X, Y, Z \} \text{ for some $k \in \{0,1,2,3\}$. } \]

This generalizes to several qubits: a Pauli operator on $n$ qubits is any operator of the form 
\[ P = i^k \bigotimes_{j=1}^n G_j, \text{ where } G_j \in \{I,X,Z,Y\}, k \in \{0,1,2,3\}. \]

We write $\mathcal{P}_n$ for the set of Pauli operators on $n$ qubits. For conciseness any Pauli operator, up to the global phase, can be written as a Pauli word, for instance $IZXI = I \otimes Z \otimes X \otimes I$. We also write $G_j, G \in \{X,Z,Y\}$, the Pauli operator that applies solely the operator $G$ on qubit $j$.

With the 2-bit encoding 
\[ I \rightarrow 00, X \rightarrow 01, Z \rightarrow 10, Y \rightarrow 11 \] 
we can always represent a Pauli operator with a boolean vector $p$ of size $2n$ where $(p[j], p[j+n])$ encodes the Pauli matrix applied on qubit $j$ plus an extra two bits to encode the value of $k$. We write $P$ interchangeably for the Pauli operator or the associated boolean vector. 

Example: 
\[ YZXI = \begin{bmatrix} 1 \\ 1 \\ 0 \\ 0 \\ --- \\ 1 \\ 0 \\ 1 \\ 0 \end{bmatrix}. \] 

Given such a boolean vector $P$, we write $P^Z$ for the first $n$ values of $P$ (encoding the $Z$ components) and $P^X$ for the last $n$ values (encoding the $X$ components). The product of two Pauli operators $(P,k) = (P_1, k_1) \times (P_2, k_2)$ is given by 
\[ P = P_1 \oplus P_2 \]
\[ k = k_1 + k_2 + \| P_2^X * P_1^Z \| - \| P_1^X * P_2^Z \|  \mod 4 \]
i.e, the boolean vector of the product of two Pauli operators is given by the bitwise XOR of the associated boolean vectors. Here $\| \cdot \|$ stands for the Hamming norm and $*$ stands for the element-wise product.

\begin{mydef}
The Clifford group on $n$ qubits is defined by
\[ \mathcal{C}_n = \{ U \in \mathcal{U}(2^n) \; | \; UPU^{\dag} \in \mathcal{P}_n \; \; \forall P \in \mathcal{P}_n \}. \]
\end{mydef}

It is well known that any Clifford operator can be generated by a quantum circuit in the gate set $\{H, S, CNOT \}$. It is common to also add the $CZ$ gate and the $\sqrt{X}$ gate in the usual gate set. A Clifford circuit is any circuit made of Clifford gates and uniquely defines a Clifford operator.

\begin{mydef}
Let $0 \leq k \leq n$ be two integers. A $k$ to $n$ Clifford isometry $V$ is defined as 
\[ V = CI_{2^n \times 2^m} \]

where $C$ belongs to the Clifford group on $n$ qubits and $I_{2^n \times 2^m}$ denotes the first $2^m$ columns of the identity matrix.
\end{mydef}

Equivalently, one can see a $k$ to $n$ Clifford isometry as a Clifford operator where $n-k$ of the inputs are set to $\ket{0}$. The case $k=0$ defines a stabilizer state and $k=n$ corresponds to a Clifford operator. An example of a $3$ to $5$ Clifford isometry defined by a Clifford circuit is given in Fig~\ref{fig::clifford_circuit}.

\begin{figure}
    \centering
    \includegraphics{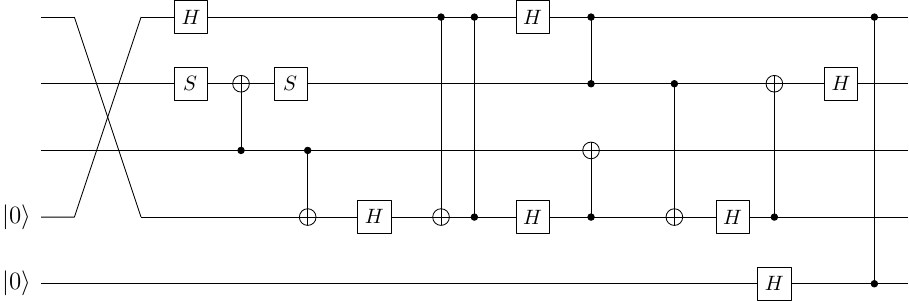}
    \caption{Example of a $5$-qubit Clifford circuit with $2$ qubits set to $\ket{0}$. This defines a $3$ to $5$ Clifford isometry.}
    \label{fig::clifford_circuit}
\end{figure}

\

For the rest of the paper, we will solely focus on the conjugation (by Clifford operators) of Pauli operators of the form 
\[ \bigotimes_{j=1}^n G_j, \text{ where } G_j \in \{I,X,Z,Y\} \]
i.e, without a global phase. 

\begin{property}

The conjugation\redsout{s} of any n-qubit Pauli word by a Clifford operator can be efficiently stored in a boolean vector of size $2n+1$.

\end{property}

\begin{proof}

The eigenvalues of the Pauli matrices without a global phase are $\{-1, 1\}$ and are unchanged after the conjugation by a Clifford. Therefore, the conjugations obtained can only be of the form 
\[ (-1)^k \bigotimes_{j=1}^n G_j, \text{ where } G_j \in \{I,X,Z,Y\}, k \in \{0,1\} \]

and we only have to store the value of $k$ on one bit.

\end{proof}

\subsection{Tableau representation of Clifford isometries}

It is well-known that Clifford operators can be efficiently represented by a polynomial sized tableau with boolean entries and the tableau can also be efficiently updated when multiplied by another Clifford operator. This relies on the fact that only the conjugations of the Pauli operators $Z_i, X_i, i=1 \hdots n$ are sufficient to fully characterize an n-qubit Clifford operator\footnote{By product we get the conjugation of any Pauli word, and the set of the $4^n$ Pauli words is a basis for the full set of complex matrices.}. The tableau is of size $(2n+1) \times 2n$, one column for each conjugation, and the $2n$ rows without the overall sign define a boolean matrix that belongs to the symplectic group $Sp(2n, \mathbb{F}_2)$. A matrix $M$ in the symplectic group satisfies the following relation:
\[ M \Omega M^T = \Omega \]

where $\Omega = \begin{bmatrix} 0 & I_n \\ I_n & 0 \end{bmatrix}$. The symplectic product between two Pauli vectors is $0$ if they commute and $1$ otherwise. So matrices in the symplectic group are simply one case of commutation relations between a set of Pauli vectors. Given that the identity tableau is symplectic and that $CP_1C^{\dag}$ and $CP_2C^{\dag}$ have the same commutation relation as $P_1$ and $P_2$, the tableau of any Clifford operator $C$ necessarily has the same commutation relations as the identity tableau and the tableau belongs to the symplectic group.

We write the tableau 
\[  M = \begin{bmatrix} Z \to Z & X \to Z \\ Z \to X & X \to X \end{bmatrix} \]
where the block $Z/X \to Z/X$ explicitly shows which part corresponds to the $Z/X$ components of the images of the $Z_i/X_i$ operators.

This efficient representation can be generalized to Clifford isometries and we give a constructive proof that a $k$ to $n$ Clifford isometry can be represented by a subset of $n+k$ columns of the tableau of a Clifford operator. We also show that this representation is not unique and we give a construction of the equivalence class.

\begin{theorem}

Any $k$ to $n$ Clifford isometry $V$ can be represented by $n+k$ columns of a Clifford tableau $M$. We write
\[ M_V = \begin{bmatrix} Z_1 \to Z & Z_2 \to Z & X \to Z \\ Z_1 \to X & Z_2 \to X & X \to X \end{bmatrix} \] 
where 
\begin{itemize}
    \item $\begin{bmatrix} Z_1 \to Z \\ Z_1 \to X \end{bmatrix} \in \mathbb{F}_2^{(2n+1) \times k}$ encodes the conjugations of $Z_i, i=1\hdots k$, 
    \item $\begin{bmatrix} Z_2 \to Z \\ Z_2 \to X \end{bmatrix} \in \mathbb{F}_2^{(2n+1) \times (n-k)}$ encodes the conjugations of $Z_i, i=k+1\hdots n$
    \item $\begin{bmatrix} X \to Z \\ X \to X \end{bmatrix} \in \mathbb{F}_2^{(2n+1) \times k}$ encodes the conjugations of $X_i, i=1\hdots k$. 
\end{itemize} 
Moreover, any of the $n-k$ columns of $\begin{bmatrix} Z_2 \to Z \\ Z_2 \to X \end{bmatrix}$ can be added to any other column other than itself such that the specification of the isometry is unchanged.

\label{thm1}

\end{theorem} 

\begin{proof}

Let $V$ be a $k$ to $n$ Clifford isometry. There exists a Clifford operator $C$ such that $V$ is characterized by the set 
\[ \{ V\ket{b} = C \ket{b}\ket{0}^{\otimes n-k} \, | \, b \in \mathbb{F}_2^k. \}. \]

Let $\ket{\psi} = C\ket{0}^{\otimes n}$, we have 
\begin{align*} \forall b \in \mathbb{F}_2^k, C \ket{b}\ket{0}^{\otimes n-k} & = C \prod_{i=1}^k X_i^{b_i} \ket{0}^{\otimes n} \\
& = C \prod_{i=1}^k X_i^{b_i} C^{\dag} C\ket{0}^{\otimes n} \\
& = \left( C \prod_{i=1}^k X_i^{b_i} C^{\dag} \right) \ket{\psi}. \end{align*}

Therefore, $V$ is completely characterized by the values of 
\[ C \prod_{i=1}^k X_i^{b_i} C^{\dag} \ket{\psi} , \forall b \in \mathbb{F}_2^k. \]

By product we only need the knowledge of $\ket{\psi}$ and $C X_i C^{\dag}, i=1 \hdots k$. Equivalently we only need the states $VX_i \ket{0}^{\otimes k}, i=1\hdots k$.

It is well-known that a Clifford state is completely determined by the knowledge of $n$ Pauli operators that stabilize the state \cite{aaronson2004improved}. Given that 
\[ Z_i \ket{0}^{\otimes n} = \ket{0}^{\otimes n}, \; i=1 \hdots n, \]
we have
\[ C Z_i C^{\dag} \ket{\psi} = \ket{\psi}, \; i=1 \hdots n \]
that represent such a set of stabilizers for $\ket{\psi}$.

Overall, we need the values of 
\[ C Z_i C^{\dag}, \; i=1 \hdots n \]
and
\[ C X_i C^{\dag}, \; i=1 \hdots k \]
to completely determine $V$. The tableau representation of $C$ gives the conjugations of $Z_i, X_i, i=1 \hdots n$ by $C$ so finally we only need the first $n+k$ columns of the tableau of $C$.

We write 
\[ M_V = \begin{bmatrix} Z_1 \to Z & Z_2 \to Z & X \to Z \\ Z_1 \to X & Z_2 \to X & X \to X \end{bmatrix} \]
where 
\begin{itemize}
    \item $\begin{bmatrix} Z_1 \to Z \\ Z_1 \to X \end{bmatrix} \in \mathbb{F}_2^{(2n+1) \times k}$ encodes the conjugations of $Z_i, i=1\hdots k$, 
    \item $\begin{bmatrix} Z_2 \to Z \\ Z_2 \to X \end{bmatrix} \in \mathbb{F}_2^{(2n+1) \times (n-k)}$ encodes the conjugations of $Z_i, i=k+1\hdots n$
    \item $\begin{bmatrix} X \to Z \\ X \to X \end{bmatrix} \in \mathbb{F}_2^{(2n+1) \times k}$ encodes the conjugations of $X_i, i=1\hdots k$. 
\end{itemize}

We now show that adding any column of $\begin{bmatrix} Z_2 \to Z \\ Z_2 \to X \end{bmatrix}$ to any column of $\begin{bmatrix} X \to Z \\ X \to X \end{bmatrix}$ does not modify the isometry specified. 

This result relies on the following identities
\begin{align*} \forall i \in \llbracket 1,k \rrbracket, \forall a \in \mathbb{F}_2^{n-k}, \; VX_i\ket{0}^{\otimes k} & = C X_i \ket{0}^{\otimes n} \\ & = C X_i \prod_{j=k+1}^n Z_j^{a_{j-k}} \ket{0}^{\otimes n} \\ 
& = C \left(X_i \prod_{j=k+1}^n Z_j^{a_{j-k}} \right) C^{\dag} \ket{\psi} \\
& = \underbrace{\left(C  X_i C^{\dag} \right)}_{\text{one column of } \begin{bmatrix} X \to Z \\ X \to X \end{bmatrix}} \times \prod_{j=k+1}^n \underbrace{\left(C  Z_j^{a_{j-k}} C^{\dag} \right)}_{\text{one column of } \begin{bmatrix} Z_2 \to Z \\ Z_2 \to X \end{bmatrix}} \ket{\psi}  , \end{align*}

Normally, to reconstruct the states $V X_i\ket{0}^{\otimes k}$ from our tableau, we need the images of the $Z_i$'s to reconstruct $\ket{\psi}$ and the images of the $X_i$. What the identity above shows is that if we know any image of a Pauli word of the form $X_i \prod_{j=k+1}^n Z_j^{a_{j-k}}, a \in \mathbb{F}_2^{n-k}$ then we can still reconstruct the desired state. This is equivalent to adding any column of $\begin{bmatrix} Z_2 \to Z \\ Z_2 \to X \end{bmatrix}$ to any column of $\begin{bmatrix} X \to Z \\ X \to X \end{bmatrix}$. Similarly, adding any column of $\begin{bmatrix} Z_2 \to Z \\ Z_2 \to X \end{bmatrix}$ to $\begin{bmatrix} Z_1 \to Z \\ Z_1 \to X \end{bmatrix}$ and $\begin{bmatrix} Z_2 \to Z \\ Z_2 \to X \end{bmatrix}$ (other than itself) only gives a new basis for the stabilizers of $\ket{\psi}$ but it does not change the specification of $\ket{\psi}$.

\end{proof}

Note that the asymmetry between the $Z$ and $X$ operators in the tableau representation of an isometry comes from the fact that we choose the ancillary qubits to be in the $\ket{0}$ state. One could have chosen to have ancillas in the state $\ket{+}$ and we would need the images of all the $X_i$'s and not all the $Z_i$'s. Any intermediate mix of $\ket{0}$ and $\ket{+}$ states as ancillas is also allowed.

Following the rules for simulating Clifford circuits by tracking the images of the Pauli operators \cite{aaronson2004improved}, we can compute the tableau of the isometry of Fig~\ref{fig::clifford_circuit}. We get 
\[ M_V = \begin{bmatrix} Z \to Z & X \to Z \\ Z\to X & X \to X \end{bmatrix} = \begin{bmatrix} A & C \\ B & D \end{bmatrix} = \begin{bmatrix} 0 & 0 & 0 & 1 & 1 && 0 & 1 & 1 \\ 0 & 1 & 1 & 1 & 0 && 0 & 0 & 1 \\ 1 & 1 & 1 & 0 & 0 && 0 & 1 & 0 \\ 1 & 0 & 1 & 1 & 0 && 0 & 0 & 0 \\ 1 & 0 & 0 & 1 & 0 && 1 & 0 & 1 \\ \\ 1 & 0 & 0 & 1 & 0 && 1 & 0 & 1 \\ 1 & 0 & 1 & 0 & 0 && 1 & 1 & 0 \\ 0 & 0 & 0 & 1 & 0 && 1 & 0 & 0 \\ 0 & 1 & 1 & 1 & 0 && 0 & 1 & 0 \\ 0 & 0 & 0 & 0 & 1 && 0 & 0 & 0 \end{bmatrix}. \]

\subsection{The synthesis problem}

We propose one formalization of the synthesis of Clifford isometries. We work in the gate set $\{H,S,CZ,CNOT\}$ which is known to be universal for Clifford computation. We suppose we are given the tableau representation $M_V$ of a $k$ to $n$ isometry. We only have access to a black box representation of the isometry in a circuit, see Figure~\ref{synthesis_circuit}. 

    \begin{figure}[h]
        \begin{adjustwidth}{-1.5cm}{}
\center
\includegraphics[scale=0.75]{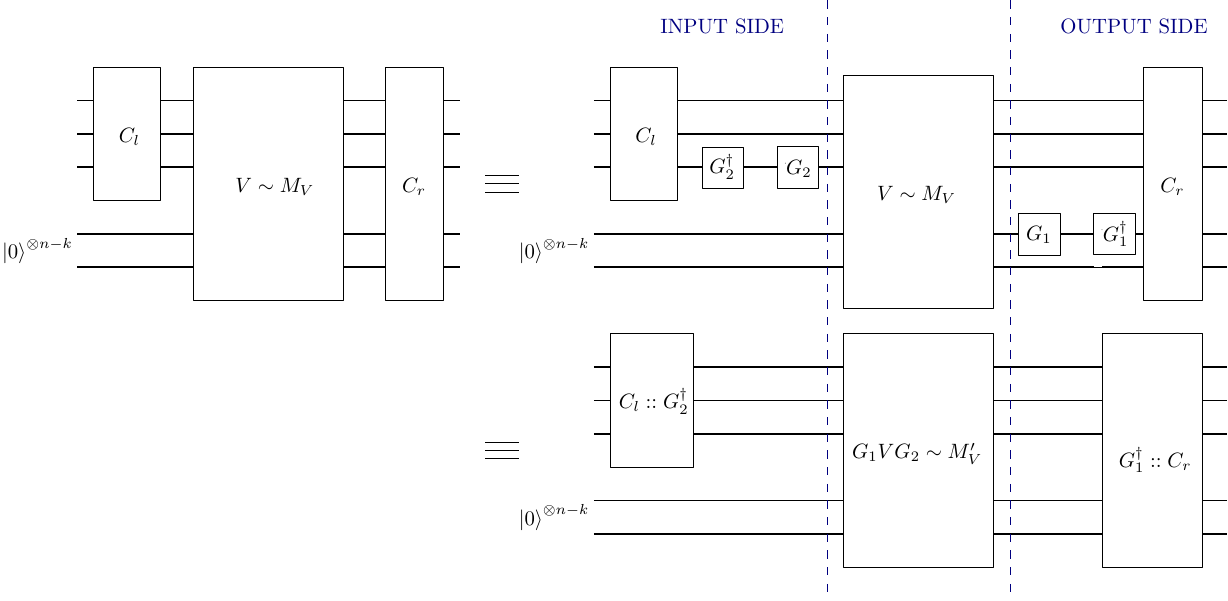}
        \end{adjustwidth}
\caption{Illustration of the synthesis framework: gates and their inverses are applied on the left and right of the isometry. One gate is merged with the isometry while the other one is stacked into a left/right circuit. Once $M_V = I$, then $C_l::C_r$ implements the desired isometry. Throughout the synthesis process, the circuit always implements the desired isometry.}
\label{synthesis_circuit}

\end{figure}

\subsubsection{Action of the elementary gates.}

\paragraph{On the output side.} Given a black box circuit implementing the isometry $V$, adding a gate $G_1$ from the right gives the operator $G_1V$ to characterize and we therefore need to compute the values of 
\[ (G_1V) \cdot Z_i \cdot (G_1V)^{\dag} = G_1\left(VZ_iV^{\dag}\right)G_1^{\dag}, i=1..n \] 
and 
\[ (G_1V)\cdot X_i \cdot (G_1V)^{\dag} = G_1\left(VX_iV^{\dag}\right)G_1^{\dag}, i=1..k. \]

The values of $VZ_iV^{\dag}$ and $VX_iV^{\dag}$ are given by $M_V$ and by the rules of conjugation we know how to directly update $M_V$. For simplicity we write 
\[ M_V = \begin{bmatrix} T \\ r \end{bmatrix} \]
where $r$ stands for the sign vector and $T$ is the $(2n \times (n+k))$ matrix that specifies the conjugations without the sign. 
We also note $\land$ for the bitwise AND operation.
\begin{itemize}
  \item Hadamard on qubit $m$: swap rows $m$ and $n+m$, set 
  \[ r = r \oplus \left(T[m,:] \land T[n+m,:] \right), \]
  \item S on qubits $m$: add row $n+m$ to row $m$, set 
  \[ r = r \oplus \left(T[m,:] \land T[n+m,:] \right), \]
  \item CZ on qubits $p, m$: add row $n+p$ to row $m$ and row $n+m$ to row $p$, set 
  \[ r = r \oplus \big[ T[n+m,:] \land T[n+p,:]  \land \left( T[p,:] \oplus T[m,:] \right) \big], \]
  \item CNOT with control $p$ and target $m$: add row $n+p$ to row $n+m$ and row $m$ to row $p$, set 
  \[ r = r \oplus \big[ T[n+p,:] \land T[m,:] \land \left( 1 \oplus T[n+m,:] \oplus T[p,:] \right)\big] . \]
\end{itemize}

\paragraph{On the input side.} When a gate $G_2$ is added from the left, we need the values of 
\[ (VG_2) \cdot Z_i \cdot (VG_2)^{\dag} \] 
and
\[ (VG_2) \cdot X_i \cdot (VG_2)^{\dag}. \] 
Writing $G_2Z_iG_2^{\dag} = \bigoplus_i a_i Z_i \oplus \bigoplus_i b_i X_i, a_i \in \{0,1\}, b_i \{0,1\}$ we have 
\[ (VG_2) \cdot Z_i \cdot (VG_2)^{\dag} = \bigoplus_i a_i \left(VZ_iV^{\dag}\right) \oplus \bigoplus_i b_i \left(VX_iV^{\dag}\right). \] 

Therefore the values of $(VG_2) \cdot Z_i \cdot (VG_2)^{\dag}$ are linear combinations of the columns of $M_V$. The same applies for the values of $(VG_2) \cdot X_i \cdot (VG_2)^{\dag}$.

\ \\

This gives:
\begin{itemize}
  \item Hadamard on qubit $m$: swap columns $m$ and $n+m$, 
  \item S on qubits $m$: add column $m$ to column $n+m$, 
  \item CZ on qubits $p, m$: add column $p$ to column $n+m$ and column $m$ to column $n+p$, 
  \item CNOT with control $p$ and target $m$: add column $p$ to column $m$ and column $n+m$ to column $n+p$, 
\end{itemize}

Similar updates for the sign vector can be derived.

For the rest of the paper we will now refer as the \textit{input qubits} the $k$ qubits on which a gate can be applied on the left. Similarly we refer as the \textit{output qubits} the $n$ qubits on which a gate can be applied on the right.

\subsubsection{The synthesis workflow.}

The synthesis process consists in adding a Clifford gate and its inverse side by side on the left or right of the black box isometry, as shown in Fig.~\ref{synthesis_circuit}. This will not change the overall functionality of the circuit but one can merge one Clifford gate to the black box and update $M_V$. At the same time the inverse gate is stacked to a list of left or right operations. When $M_V$ is equal to the identity tableau, we have a valid sequence of operations implementing the Clifford isometry by stacking the left operations and the \textit{reverse} of the right operations. 

Note that we can only focus on reducing $T$ to the identity because the sign vector can be treated separately at the end of the computation. Indeed, if $T = I$, adding twice an $S$ gate on qubit $m$ at the end of the circuit will simply do the operation 
\[ r[m] = r[m] \oplus 1. \]
In other words, keeping track of the value of $r$ is sufficient because it can be ultimately zeroed with a layer of $Z$ gates. 

\

Designing synthesis algorithms directly from the tableau representation is clumsy. For the synthesis of Clifford states and Clifford operators the derivation of normal forms is common instead of proposing a direct synthesis algorithm \cite{aaronson2004improved, maslov2018shorter, duncan2020graph, bravyi2021hadamard, bataille2021reduced,van2010classical, garcia2014geometry}. 

We try to overcome both the difficulty of manipulating the Clifford tableaus and the necessity to use normal forms by proposing a simple framework for the synthesis of Clifford isometries. One advantage of our framework will be its versatility: it can be used to design either direct synthesis algorithms or to prove normal forms.

\section{A graph-state formalism for Clifford isometries synthesis}

We prove that any $k$ to $n$ Clifford isometry can be put in a graph-state form up to local Hadamard gates. 
Then we derive a synthesis framework by showing the behavior of the different elementary Clifford gates. 

\subsection{Graph-state structure of the tableau representation}

First, we prove a normal form for symplectic matrices that exposes an underlying graph-state structure of Clifford operators. 

\begin{theorem}
\label{theorem_clifford}
Let $M = \begin{bmatrix} A & C \\ B & D \end{bmatrix} = \begin{bmatrix} Z \to Z & X \to Z \\ Z \to X & X \to X \end{bmatrix} \in Sp(2n, \mathbb{F}_2)$. If $B$ is invertible, then there exist two boolean symmetric matrices $S, S'$ such that 
\begin{equation} M = \left[ \begin{array}{c|c} S'B & (B^T)^{-1} \oplus S'BS \\ \hline B & BS \end{array} \right].  \label{sympl_form} \end{equation}
\end{theorem}

\begin{proof}

If $B$ is invertible, then we can define $S = B^{-1}D$ and $S' = AB^{-1}$ and $K = C \oplus (B^T)^{-1}$ such that 
\[ A = S'B,  \]
\[ D = BS, \]
\[ C = (B^T)^{-1} \oplus K. \]

Matrices in the symplectic group verify three different equations. One of them will prove that $K = S'BS$ and the two others will show that $S$ and $S'$ are symmetric. 

\begin{enumerate}
  \item $B^TA \oplus A^TB = 0$ gives 
  \[ B^TS'B \oplus B^TS'^TB = B^T(S' \oplus S'^T)B = 0 \] 
  which is possible if and only if $S' = S'^T$. \ \\
  \item $B^TC \oplus A^TD = I$ gives 
  \[ B^T((B^T)^{-1} \oplus K) \oplus B^TS'^TBS = I \oplus B^T(K \oplus S'BS) = I. \] 
  So we have $B^T(K \oplus S'BS) = 0$ which is possible if and only if $K = S'BS$. \ \\
  \item $D^TC \oplus C^TD = 0$ gives 
  \begin{align*} S^TB^T(B^{-T} \oplus S'BS) \oplus ((B^{-T} \oplus S'BS))^TBS & = S^TB^TB^{-T} \oplus S^TB^TS'BS \oplus B^{-1}BS \oplus S^TB^TS'^TBS \\ &= S^T \oplus S = 0 \end{align*}
  which is possible if and only if $S = S^T$.
\end{enumerate}

\end{proof}

\ \\
This trivially extends to Clifford isometries: 
\begin{corollary}
\label{theorem_clifford}
Let $M_V = \begin{bmatrix} A & C \\ B & D \end{bmatrix} \in \mathbb{F}_2^{2n \times (n+k)}$ be the tableau representation of a Clifford isometry. If $B$ is invertible, then there exist two boolean symmetric matrices $G_k \in F_2^{k \times k}, G_n \in F_2^{n \times n}$ such that 
\begin{equation} M_V = \left[ \begin{array}{c|c} G_nB & (B^T)^{-1} \begin{pmatrix} I_k \\ 0 \end{pmatrix} \oplus G_nB \begin{pmatrix} G_k \\ 0 \end{pmatrix} \\ \hline B & B \begin{pmatrix} G_k \\ 0 \end{pmatrix} \end{array} \right]. \label{isometry_form} \end{equation}
\end{corollary}

\begin{proof}

From Eq.~\ref{sympl_form}, by taking the first $n+k$ columns we have an operator of the form 
\[ M_V = \left[ \begin{array}{c|c} G_nB & (B^T)^{-1} \begin{pmatrix} I_k \\ 0 \end{pmatrix} \oplus G_nB \begin{pmatrix} G_k \\ F \end{pmatrix} \\ \hline B & B \begin{pmatrix} G_k \\ F \end{pmatrix} \end{array} \right] \]
where $F \in F_2^{(n-k) \times k}$ is an arbitrary rectangular boolean matrix, $G_n$ is $S'$ in Eq.~\ref{sympl_form} and $G_k$ is the top block $k \times k$ of $S$. 

We rewrite $M_V$ as 
\[ M_V = \left[ \begin{array}{c|c|c} G_nB \begin{pmatrix} I_k \\ 0_{(n-k) \times k} \end{pmatrix} & G_nB \begin{pmatrix} 0_{k \times (n-k)} \\ I_{n-k} \end{pmatrix} & (B^T)^{-1} \begin{pmatrix} I_k \\ 0 \end{pmatrix} \oplus G_nB \begin{pmatrix} G_k \\ F \end{pmatrix} \\ \hline B \begin{pmatrix} I_k \\ 0_{(n-k) \times k} \end{pmatrix} & B \begin{pmatrix} 0_{k \times (n-k)} \\ I_{n-k} \end{pmatrix} & B \begin{pmatrix} G_k \\ F \end{pmatrix} \end{array} \right] \]

by splitting the first $n$ columns of $M_V$ and we recall from Theorem~\ref{thm1} that adding any column $k+1 \hdots n$ to any column $n+1 \hdots n+k$ of $M_V$ is free and does not change the functionality of the Clifford isometry. So with suitable column operations we can zero $F$. Hence the result.

\end{proof}

\begin{theorem}

Any Clifford isometry, up to Hadamard gates on the output qubits, has a normal form given by Eq.~\ref{isometry_form}.

\end{theorem}

\begin{proof}

It is sufficient to prove that for any tableau representation of a Clifford isometry the action of Hadamard gates on the output qubits, i.e, specific row swaps on the tableau, is enough to make $B$ invertible. This is a well-known result already proved in \cite{aaronson2004improved} for instance.

\end{proof}

\ \\

At first sight, it seems that we need to track the values of $G_n, G_k, B$ to perform the synthesis of a Clifford isometry. In fact, the knowledge of $G_n, G_k $ and a submatrix of $B$ is sufficient: 

\begin{theorem}

Given a Clifford isometry of the form given by Eq.~\ref{isometry_form}, if $G_k, G_n$ are equal to the null matrix and if $B = \begin{bmatrix} I_k & 0_{k \times (n-k)} \\ B' & B'' \end{bmatrix}$ for some boolean matrices $B', B''$, then the Clifford isometry acts as the $H^{\otimes n}$ operator on the quantum memory. 

\end{theorem}

\begin{proof}

Replacing in Eq.~\ref{isometry_form} gives 
\[ M_V = \left[ \begin{array}{c|c} 0 & (B^T)^{-1} \begin{pmatrix} I_k \\ 0 \end{pmatrix} \\ \hline B & 0 \end{array} \right]. \]

Given that $B = \begin{bmatrix} I_k & 0_{k \times (n-k)} \\ B' & B'' \end{bmatrix}$, we also have $B^{-1} = \begin{bmatrix} I_k & 0_{k \times (n-k)} \\ B''' & B'''' \end{bmatrix}$ for some boolean matrices $B''', B''''$. Therefore, it is easy to check that 
\[ (B^T)^{-1} \begin{bmatrix} I_k \\ 0 \end{bmatrix} = \begin{bmatrix} I_k \\ 0 \end{bmatrix} \]
Overall, we have a new block structure for $M_V$, 
\[ M_V = \left[ \begin{array}{c|c|c} 0 & 0 & \begin{pmatrix} I_k \\ 0 \end{pmatrix} \\ \hline I_k & 0_{k \times (n-k)} & 0 \\ B' & B'' & 0 \end{array} \right]. \]

$B$ is invertible, therefore $B''$ is also invertible. We recall that we can add any column of $B''$ on any other column of $M_V$, it does not change the functionality of the Clifford isometry. So with suitable column operations we can reduce $B''$ to the identity operator and zero $B'$. We end with the operator 

\[ M_V = \left[ \begin{array}{c|c} 0_{k \times n} & I_k \\ 0_{(n-k) \times n}  & 0_{(n-k) \times k} \\ \hline  I_n & 0_{n \times k} \end{array} \right]. \]

which correspond to the $k$ to $n$ isometry $H^{\otimes k}\otimes\ket{+}^{n-k}$.

\end{proof}

This means that the last $n-k$ rows of $B$ or $B^{-1}$ do not influence the specification of the identity isometry and it is useless to keep track of their values in the synthesis process. Instead, we set 
\[ G = \begin{bmatrix} G_k & B_{k,n} \\ B_{k,n}^T & G_n \end{bmatrix} \]
where $B_{k,n} = B^{-1}[1:k,:] \in F_2^{k \times n}$. This is the \textit{graph-state form} of the Clifford isometry because $G$ can be seen as an adjacency matrix and as we will see the $CZ$ gates directly add or remove an edge from $G$, similarly to the graph-state formalism. We illustrate the computational process on the example of the isometry in Fig~\ref{fig::clifford_circuit}. We recall that we computed its tableau 
\[ M_V = \begin{bmatrix} A & C \\ B & D \end{bmatrix} = \begin{bmatrix} 0 & 0 & 0 & 1 & 1 && 0 & 1 & 1 \\ 0 & 1 & 1 & 1 & 0 && 0 & 0 & 1 \\ 1 & 1 & 1 & 0 & 0 && 0 & 1 & 0 \\ 1 & 0 & 1 & 1 & 0 && 0 & 0 & 0 \\ 1 & 0 & 0 & 1 & 0 && 1 & 0 & 1 \\ \\ 1 & 0 & 0 & 1 & 0 && 1 & 0 & 1 \\ 1 & 0 & 1 & 0 & 0 && 1 & 1 & 0 \\ 0 & 0 & 0 & 1 & 0 && 1 & 0 & 0 \\ 0 & 1 & 1 & 1 & 0 && 0 & 1 & 0 \\ 0 & 0 & 0 & 0 & 1 && 0 & 0 & 0 \end{bmatrix}. \]
One can check that $B$ is already full rank. We have 
\[ B^{-1} = \begin{bmatrix} 1 & 0 & 1 & 0 & 0 \\ 1 & 1 & 0 & 1 & 0 \\ 1 & 1 & 1 & 0 & 0 \\ 0 & 0 & 1 & 0 & 0 \\ 0 & 0 & 0 & 0 & 1 \end{bmatrix} \]
and from it follows the values of $G_n, G_k$ and $B_{kn}$: 
\[ G_n = AB^{-1} = \begin{bmatrix} 0 & 0 & 1 & 0 & 1 \\ 0 & 0 & 0 & 1 & 0 \\ 1 & 0 & 0 & 1 & 0 \\ 0 & 1 & 1 & 0 & 0 \\ 1 & 0 & 0 & 0 & 0 \end{bmatrix}, \]
\[ B_{kn} = \begin{pmatrix} I_k & 0_{k \times (n-k} \end{pmatrix} B^{-1} = \begin{bmatrix} 1 & 0 & 1 & 0 & 0 \\ 1 & 1 & 0 & 1 & 0 \\ 1 & 1 & 1 & 0 & 0 \end{bmatrix}. \]
\[ G_k = B_{kn}D = \begin{bmatrix} 0 & 0 & 1 \\ 0 & 0 & 1 \\ 1 & 1 & 1 \end{bmatrix},  \]

The goal is to reduce $G$ to 
\[ \left[ \begin{array}{c|cc} 0_k & I_k & 0_{k \times (n-k)} \\ \hline I_k & \multicolumn{2}{c}{\multirow{2}{*}{$0_{n \times n}$}} \\ 0_{n-k \times k} \end{array} \right]. \] 
Indeed, when $G_n=0, G_k=0$ and $B_{k,n} = \begin{bmatrix} I_k \\ 0_{(n-k) \times k} \end{bmatrix}$ we have seen from the proof above that $M_V$ acts as the $H^{\otimes n}$ operator. Adding $n$ Hadamard gates on the output side in the framework illustrated in Fig.~\ref{synthesis_circuit} reduces $M_V$ to the identity isometry and this finishes the synthesis.

\subsection{Available operations} \label{sec::ope}

We now list the action of the authorized operations on $G$. We write $e_i$ the $i$-th canonical vector, $e_{ij}$ the matrix with all zero except on entry $(i,j)$ and $E_{ij} = I \oplus e_{ij}$. For clarity we also note $\tilde{e}_{ij} = e_{ij} \oplus e_{ji}$ if $i \neq j$ and $\tilde{e}_{ii} = e_{ii}$. The proofs are given in the Appendix.
\begin{itemize}
  \item $S$ on the output qubit $i$ ($1 \leq i \leq n$): 
  \[ G \leftarrow G \oplus \tilde{e}_{i+k, i+k} \]
  \item $S$ on the input qubit $i$ ($1 \leq i \leq k$):
  \[ G \leftarrow G \oplus \tilde{e}_{i, i} \]
  \item $CZ$ on the output qubits $i,j$ ($1 \leq i < j \leq n$):
  \[ G \leftarrow G \oplus \tilde{e}_{i+k, j+k} \]
  \item $CZ$ on the input qubits $i,j$ ($1 \leq i < j \leq n$):
  \[ G \leftarrow G \oplus \tilde{e}_{ij} \]
  \item $CNOT$ on the output qubits with control $i$, target $j$ ($1 \leq i < j \leq n$): 
  \[ G \leftarrow E_{i+k,j+k}GE_{j+k,i+k}. \]
  \item $CNOT$ on the input qubits with control $i$, target $j$ ($1 \leq i < j \leq k$): 
  \[ G \leftarrow E_{ij}GE_{ji}. \]
\end{itemize}
Contrary to the other gates that preserve the structure of Eq.~\ref{isometry_form}, the Hadamard gate and the $R_x(\pi/2)$ gate do not guarantee it because the matrix $B$ can become non invertible. We need to separate two cases.
\begin{itemize}  
  \item if $G_n[i,i] = 0$, then we can apply an $R_x(\pi/2)$ on the output qubit $i$:
  \[ G \leftarrow G \oplus G[:,k+i]G[:,k+i]^T. \]
  \item if $G_n[i,i] = 1$, then we apply an Hadamard gate on the output qubit $i$, or equivalently we apply $S R_x(\pi/2) S$, and the action on $G$ is 
  \[ G \leftarrow G \oplus g_ig_i^T \]
  where $g_i = G[:,k+i] \oplus e_i$. 
  \item if $G_k[i,i] = 0$, then we can apply an $R_x(\pi/2)$ on the input qubit $i$, we get 
  \[ G \leftarrow G \oplus G[:,i]G[:,i]^T. \]
  \item if $G_k[i,i] = 1$, then we apply an Hadamard gate on the input qubit $i$, or equivalently we apply $S R_x(\pi/2) S$, and the action on $G$ is 
  \[ G \leftarrow G \oplus g_ig_i^T \]
  where where $g_i = G[:,i] \oplus e_i$. 
\end{itemize}

The available operations are summarized in Table~\ref{framework}. For practical reasons, we also give the available operations as actions on the submatrices $G_n, G_k$ and $B_{k,n}$ in Table~\ref{framework2}. We will mostly design algorithms in the gate set \{CNOT, CZ, S\}, without the Hadamard gate, and it is easier to give intuitions on their behavior by looking at the action of the gates on $G_n, G_k$ and $B_{k,n}$ separately rather than on the full graph matrix $G$. Notably, we want to emphasize the following key points: 
\begin{itemize}
    \item CZ gates and S gates only flip an entry of either $G_k$ or $G_n$, depending on which side those gates are applied, 
    \item a CNOT circuit on the output side will apply an operator $C$ on $B_{k,n}^T$, giving the new operator $CB_{k,n}^T$ and will conjugate $G_n$ by the same operator, giving $CG_nC^T$,
    \item similarly a CNOT circuit on the input side will apply an operator $C$ on $B_{k,n}^T$, giving the new operator $B_{k,n}^TC$ and will conjugate $G_k$ by the same operator, giving $C^TG_kC$.
\end{itemize}

\subsection{Summary: the synthesis framework}

We have shown that any $k$ to $n$ Clifford isometry, up to Hadamard gates, is completely determined by a symmetric boolean matrix $G \in F_2^{(n+k) \times (n+k)}$. $G$ has the following block structure 
\[ G = \begin{bmatrix} G_k & B_{k,n} \\ B_{k,n}^T & G_n \end{bmatrix} \] 
where $G_k \in F_2^{k \times k}, G_n \in F_2^{n \times n}$ are symmetric and $B_{k,n} \in F_2^{k \times n}$ is full rank. Synthesizing the isometry is equivalent to reducing $G$ to the identity operator 
\[ \left[ \begin{array}{c|cc} 0_k & I_k & 0_{k \times n-k} \\ \hline I_k & \multicolumn{2}{c}{\multirow{2}{*}{$0_{n \times n}$}} \\ 0_{n-k \times k} \end{array} \right] \] 
with the available operations given in Tables~\ref{framework} and ~\ref{framework2}.

\begin{table}
\resizebox{1.1\columnwidth}{!}{
\begin{tabular}{ccccc}
\toprule
\toprule
Clifford gate & Action on $G$ & Matrix interpretation & Graph interpretation & Range of validity \\
\cmidrule(lr){1-1} \cmidrule(lr){2-2} \cmidrule(lr){3-3} \cmidrule(lr){4-4} \cmidrule(lr){5-5}
$S_i$ & $G[i,i] \leftarrow G[i,i] \oplus 1$ & Flips diagonal entry & Complements loop of node $i$ & $ 1 \leq i \leq n+k$ \\
\\
\multirow{2}{*}{$CZ_{i,j}$} & $G[i,j] \leftarrow G[i,j] \oplus 1$ & Flips two symmetric &\multirow{2}{*}{Complements edge $(i,j)$} & $ 1 \leq i, j \leq k$ \\
& $G[j,i] \leftarrow G[j,i] \oplus 1$ &entries&& or $ k+1 \leq i, j \leq n+k$ \\
\\
\multirow{2}{*}{$CNOT_{i,j}$} & $G[i,:] \leftarrow G[i,:] \oplus G[j,:]$ & Elementary row & Complements the neighborhood & $ 1 \leq i, j \leq k$  \\
& $G[:,i] \leftarrow G[:,i] \oplus G[:,j]$ & and column operation & of node $j$ for node $i$ & or $ k+1 \leq i, j \leq n+k$ \\
\\
$H_i$ & $G \leftarrow G \oplus gg^T $ & Rank-one update & Local complementation on node $i$ & $ 1 \leq i \leq n+k$ \\
(if $G[i,i]=1$)  & $g = G[:,i] \oplus e_i$ && (with loops except on node $i$) \\
\\
$(R_x)_i$ & $G \leftarrow G \oplus gg^T$ & Rank-one update & Local complementation on node $i$ & $ 1 \leq i \leq n+k$  \\
(if $G[i,i]=0$) & $g = G[:,i]$  && (with loops) \\
\bottomrule
\end{tabular}}
\caption{Available operations for the synthesis of a $k$ to $n$ Clifford isometry $G$.}
\label{framework}
\end{table}

\begin{table}
\centering
\begin{tabular}{cccccl}
\toprule
\toprule
Clifford gate & Side & Qubits & Acts on & & Action \\
\cmidrule(lr){1-1} \cmidrule(lr){2-2} \cmidrule(lr){3-3} \cmidrule(lr){4-4} \cmidrule(lr){6-6}
$S$ & Input & $i$ & $G_k$ && $G_k[i,i] \leftarrow G_k[i,i] \oplus 1$ \\
\\
$S$ & Output & $i$ & $G_n$ && $G_n[i,i] \leftarrow G_n[i,i] \oplus 1$ \\
\\
\multirow{2}{*}{$CZ$} & \multirow{2}{*}{Input} & \multirow{2}{*}{$i,j$} & \multirow{2}{*}{$G_k$} && $G_k[i,j] \leftarrow G_k[i,j] \oplus 1$ \\
&&&&& $G_k[j,i] \leftarrow G_k[j,i] \oplus 1$ \\
\\
\multirow{2}{*}{$CZ$} & \multirow{2}{*}{Output} & \multirow{2}{*}{$i,j$} & \multirow{2}{*}{$G_n$} && $G_n[i,j] \leftarrow G_n[i,j] \oplus 1$ \\
&&&&& $G_n[j,i] \leftarrow G_n[j,i] \oplus 1$ \\
\\
\multirow{3}{*}{$CNOT$} & \multirow{3}{*}{Input} & \multirow{3}{*}{$i,j$} & \multirow{3}{*}{$G_k, B_{k,n}$} && $G_k[i,:] \leftarrow G_k[j,:] \oplus G_k[i,:]$ \\
&&&&& $G_k[:,i] \leftarrow G_k[:,j] \oplus G_k[:,i]$ \\
&&&&& $B_{k,n}[i,:] \leftarrow B_{k,n}[j,:] \oplus B[i,:]$ \\
\\
\multirow{3}{*}{$CNOT$} & \multirow{3}{*}{Output} & \multirow{3}{*}{$i,j$} & \multirow{3}{*}{$G_n, B_{k,n}$} && $G_n[i,:] \leftarrow G_n[j,:] \oplus G_n[i,:]$ \\
&&&&& $G_n[:,i] \leftarrow G_n[:,j] \oplus G_n[:,i]$ \\
&&&&& $B_{k,n}[:,i] \leftarrow B_{k,n}[:,j] \oplus B[:,i]$ \\
\\
\multirow{5}{*}{$H/R_x$} & \multirow{5}{*}{Input} & \multirow{5}{*}{$i$} & \multirow{5}{*}{$G_k, G_n, B_{k,n}$} && $G_k \leftarrow G_k \oplus gg^T$ \\
&&&&& $G_n \leftarrow G_n \oplus B_{k,n}[i,:]^TB_{k,n}[i,:]$ \\
&&&&& $B_{k,n} \leftarrow B_{k,n} \oplus gB_{k,n}[i,:]$ \\
&&&&& ($g = G_k[:,i]$ if $G_k[i,i]=0$, \\
&&&&& $g = G_k[:,i] \oplus e_i$ otherwise) \\
\\
\multirow{5}{*}{$H/R_x$} & \multirow{5}{*}{Output} & \multirow{5}{*}{$i$} & \multirow{5}{*}{$G_k, G_n, B_{k,n}$} && $G_n \leftarrow G_n \oplus gg^T$ \\
&&&&& $G_k \leftarrow G_k \oplus B_{k,n}[:,i]^TB_{k,n}[:,i]$ \\
&&&&& $B_{k,n} \leftarrow B_{k,n} \oplus B_{k,n}[:,i]g^T$ \\
&&&&& ($g = G_n[:,i]$ if $G_n[i,i]=0$, \\
&&&&& $g = G_n[:,i] \oplus e_i$ otherwise) \\
\bottomrule
\end{tabular}
\caption{An alternative way to present the available operations for the synthesis of a $k$ to $n$ Clifford isometry $G = [G_k \; B_{k,n} ; B_{k,n}^T \; G_n]$. }
\label{framework2}
\end{table}

Overall, everything acts as if we have a graph state on $n+k$ qubits on which only a subset of operations are available: two-qubit gates are only possible among one set of $k$ qubits or the complementary set of $n$ qubits. Contrary to the graph state synthesis where one wants to remove every edge of the graph, we look for an "identity" layout of the edges where $k$ qubits of the first set are solely connected to one qubit of the other set. This is illustrated in Figure~\ref{fig::graph} with the application of one CZ and one CNOT gate. This form is a generalization of the form derived in \cite{duncan2020graph} but we did not use the ZX-calculus in our proof and we provide a more in-depth characterization of the synthesis framework implied by this graph-state form.

We remind that our framework is not universal in the sense that we are not authorized to apply any Clifford gate, therefore we cannot generate any Clifford circuit. Solving optimally an instance of our framework does not guarantee that we have an optimal implementation of the Clifford isometry.

We now review the possibilities offered by this new framework and propose some algorithms to perform the synthesis of Clifford isometries.

\begin{figure}

\center
\subfloat[Each submatrix $G_n, G_k, B_{k,n}$ corresponds to adjacency relationships between nodes of the graph.]{\includegraphics[scale=0.9]{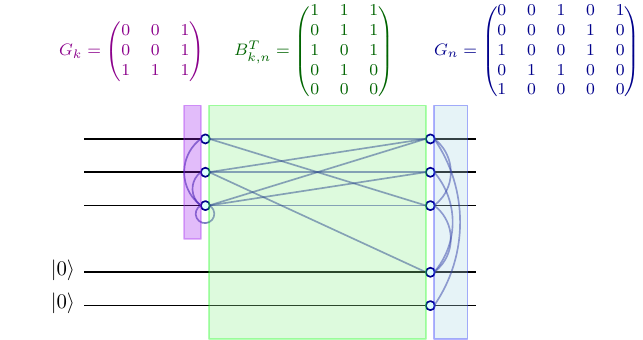}} 
\\
\hspace*{-2cm}
\subfloat[Similarly to the workflow explained in Fig.~\ref{synthesis_circuit}, one can add a gate and its inverse on the input or output side without changing the functionality of the global circuit. One gate is merged to a circuit and the other is absorbed by the graph state which is modified accordingly. In red the entries, rows or columns that are modified by the gate. The corresponding added edges are also in red, removed edges are red and dotted.]{\includegraphics[scale=0.9]{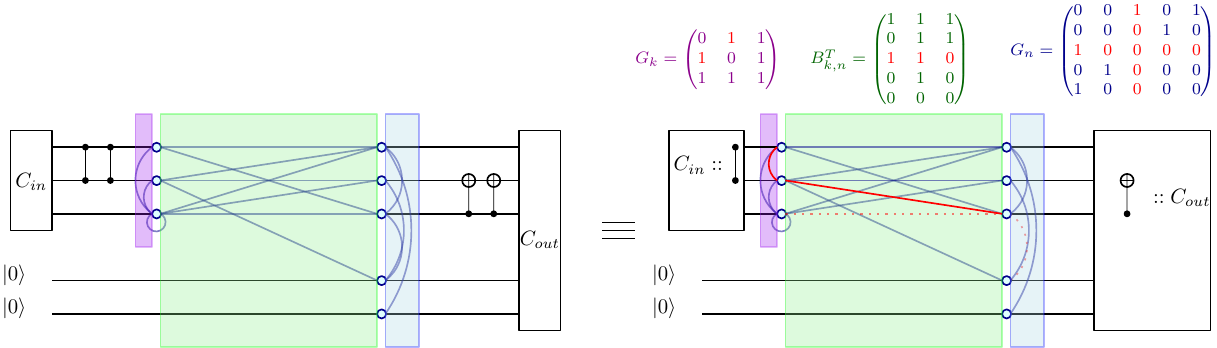}}
\\
\subfloat[Once the identity graph is obtained, the isometry is equivalent to a layer of Hadamard gates and this ends the synthesis.]{\includegraphics[scale=0.9]{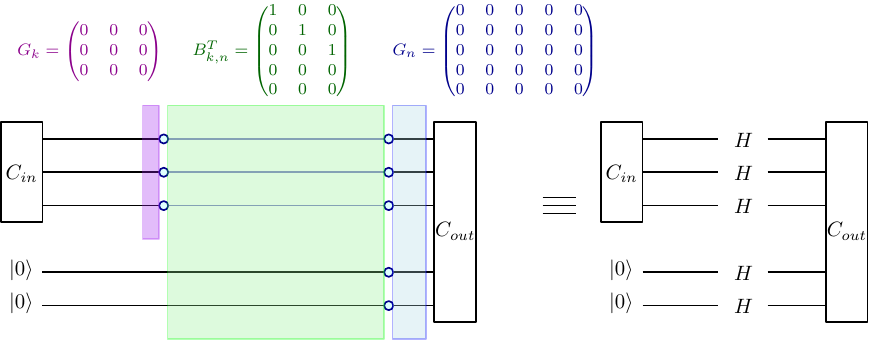}}
\caption{Illustration of the graph-state framework.}
\label{fig::graph}
\end{figure}

\section{Synthesis algorithms}

\subsection{Normal forms} \label{normal_forms}

One advantage of this framework is that normal forms can be derived quite easily with short proofs. We now propose some of them. 
Notably, we show that we can recover similar structures to some previous normal forms found in the literature. By similar structures we mean that the number and properties of the 2-qubit gate layers are the same, but our own normal forms may differ in the 1-qubit gate layers. Given that the main cost in the implementation of Clifford isometries lie in the execution of the 2-qubit gate layers, we believe that our own normal forms can be assimilated as the ones of the literature despite some differences in the 1-qubit gate layers.

We use the following terminology: -S- corresponds to a layer of S gates, -C- corresponds to a layer of CNOT gates, -CZ- corresponds to a layer of CZ gates and -H- corresponds to a layer of H gates. Given that the gates on the input side are only applied to the first $k$ qubits, and that the gates in the output side can be applied to the whole $n$ qubits, we will note the layers with a subscript to emphasize when it concerns only $k$ or $n$ qubits.

Note that our normal forms always end with a layer -H\textsubscript{n}- of Hadamard gates: this layer transforms any isometry into its graph-state form. For each qubit, there may or may not be an Hadamard gate applied. We also always have a full layer -H- of $n$ Hadamard gates "in the middle" of our normal forms that ends the synthesis of our graph-state framework. This layer draws a boundary between the operations on the input qubits and the operations on the output qubits. For clarity we omit the subscript for this layer.

Note also that the layers applied on the output side have to be read in reverse order to understand their action in our graph-state framework.

\subsubsection{S\textsubscript{k}-CZ\textsubscript{k}-H-C\textsubscript{n}-CZ\textsubscript{n}-S\textsubscript{n}-H\textsubscript{n}.}

The layers of -S- and -CZ- gates on the input side zero $G_k$. The layers -S- and -CZ- on the output side zero $G_n$. Finally the -C- layer reduces $B$ to the identity. Note that we are free to place these three last layers (CNOT, S and CZ) as we like: this will not change the -C- layer but we have to adjust the -S- and -CZ- layers depending on the value of $G_n$ that can be changed by the action of the CNOT layer. In fact, as long as a suitable phase polynomial is implemented we have a valid implementation of the Clifford isometry.

With a full Clifford operator, i.e when $k=n$, we recover the normal form -H-S-CZ-C-H-CZ-S-H- given in \cite{duncan2020graph}. Note that, in this particular case, the -C- layer can be placed before or after the -H- layer: the only difference is that we will act on the $B$ or $B^T$ part of $G$. We also have one -H- layer less.

\begin{corollary}

Given a $k$ to $n$ Clifford isometry $V$ under graph-state form $\begin{bmatrix} G_k & B_{k,n} \\ B_{k,n}^T & G_n \end{bmatrix}$, then, up to Pauli operators, 
\begin{equation} V = \frac{1}{\sqrt{2}^n} \sum_{x \in \{0,1\}^k} \sum_{y \in \{0,1\}^n} e^{i \frac{\pi}{2} \left(x^T G_k x + y^TG_ny \right) + i\pi x^T B_{k,n} y} \ket{y} \bra{x} \label{iso_eq} \end{equation}

\end{corollary}

\begin{proof}

We start from the identity isometry and we check that the normal form gives the expected formula. We write $I_{n,k}$ the $n \times n$ matrix with the first $k$ columns equal to the identity, and $0$ everywhere else.

\begin{align*} \sum_{x \in \{0,1\}^k} \ket{x}\ket{0}^{\otimes n-k}\bra{x} & \xrightarrow[]{CZ+S} & \sum_{x \in \{0,1\}^k} e^{i \frac{\pi}{2} x^T G_k x} \ket{x}\ket{0}^{\otimes n-k}\bra{x} \\
& \xrightarrow[]{H^{\otimes n}} & \frac{1}{\sqrt{2}^n}\sum_{x \in \{0,1\}^k} \sum_{y \in \{0,1\}^n} e^{i \frac{\pi}{2} x^T G_k x + i\pi \sum_{j=1}^k x_jy_j} \ket{y} \bra{x} \\
& \xrightarrow[]{CNOT} & \frac{1}{\sqrt{2}^n}\sum_{x \in \{0,1\}^k} \sum_{y \in \{0,1\}^n} e^{i \frac{\pi}{2} x^T G_k x + i\pi x^T I_{n,k} y} \ket{B y} \bra{x} \\
& \xrightarrow[]{CZ+S} & \frac{1}{\sqrt{2}^n}\sum_{x \in \{0,1\}^k} \sum_{y \in \{0,1\}^n} e^{i \frac{\pi}{2} \left(x^T G_k x + y^TB^TG_nBy \right) + i\pi x^T I_{n,k} y} \ket{B y} \bra{x} \\
& \xrightarrow[]{y \leftarrow B y} & \frac{1}{\sqrt{2}^n}\sum_{x \in \{0,1\}^k} \sum_{y \in \{0,1\}^n} e^{i \frac{\pi}{2} \left(x^T G_k x + y^TG_ny \right) + i\pi x^T I_{n,k} B^{-1} y} \ket{y} \bra{x} \\
& \xrightarrow[]{B_{k,n} = B^{-1}[1:k,:] = I_{n,k}B^{-1}} & \frac{1}{\sqrt{2}^n}\sum_{x \in \{0,1\}^k} \sum_{y \in \{0,1\}^n} e^{i \frac{\pi}{2} \left(x^T G_k x + y^TG_ny \right) + i\pi x^T B_{k,n} y} \ket{y} \bra{x} \\
\end{align*}

\end{proof}

\subsubsection{S\textsubscript{k}-C\textsubscript{k}-S\textsubscript{k}-H-C\textsubscript{n}-S\textsubscript{n}-C\textsubscript{n}-S\textsubscript{n}-H\textsubscript{n}.}

We can remove any use of -CZ- layers with the use of CNOT gates. Any symmetric boolean matrix $M$ can be written $M = LDL^T \oplus D'$ where $L$ is lower triangular and $D, D'$ are diagonal \cite{aaronson2004improved}. This applies to $G_k$ and $G_n$. Write $G_k = LDL^T \oplus D'$. The first layer of -S- gates on the input side adds $D'$ to $G_k$ such that we get $G_k \oplus D' = LDL^T$. Then, referring to the action of CNOT circuits in Section~\ref{sec::ope}, the -C- layer only needs to implement the triangular operator $L^{-1}$ on the output side to reduce $G_k \oplus D'$ to $D$ which is zeroed with a second layer of -S- gates. We have a similar argument to treat the zeroing of $G_n$. The last layer -C- on the output side ultimately reduces $B$ to the identity.

For the full Clifford operator case, we recover a normal form similar to -C-S-C-S-H-C-S-C-S- given in \cite{maslov2018shorter}.
In this case each -C- stage implements a triangular operator. This is also true in our own normal form as two -C- blocks are already triangular and the -C- block that reduces $B$ to the identity can be written as a product of two triangular matrices.

\subsubsection{Example}

We illustrate the two normal forms given above with the concrete example of Fig.~\ref{fig::graph}. The exact circuits are given and for each layer we show its effect on the graph matrix $G$ and how we eventually recover the identity. This is presented in Figures~\ref{fig::normal_form1} and \ref{fig::normal_form2}.

\begin{figure}

\center

\includegraphics[scale=1]{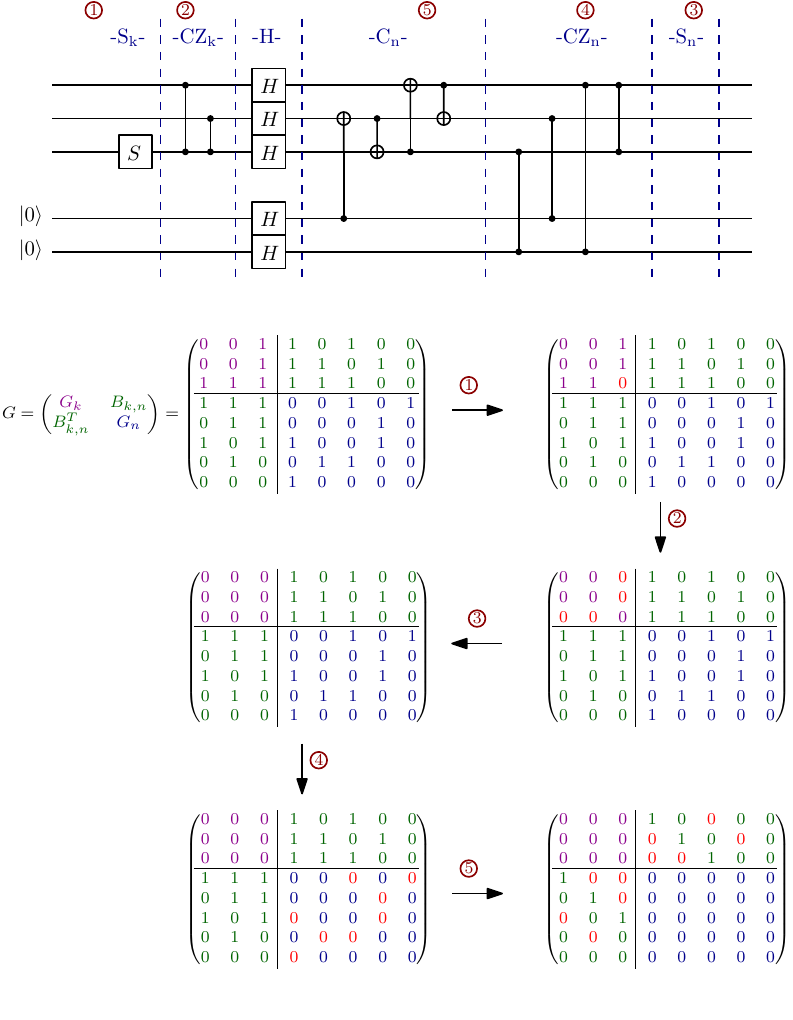}

\caption{Normal form S\textsubscript{k}-CZ\textsubscript{k}-H-C\textsubscript{n}-CZ\textsubscript{n}-S\textsubscript{n}-H\textsubscript{n} for the Clifford isometry given in Fig.~\ref{fig::graph}. We show how each layer acts on the original graph state matrix to eventually reduce it to the identity. To do so, we recall that we read the circuit on the input side from left to right and the circuit on the output side from right to left, as explained in Figures \ref{synthesis_circuit} and \ref{fig::graph}.}
\label{fig::normal_form1}
\end{figure}

\begin{figure}
\vspace*{-2cm}
\hspace*{-2cm}
\includegraphics[scale=0.9]{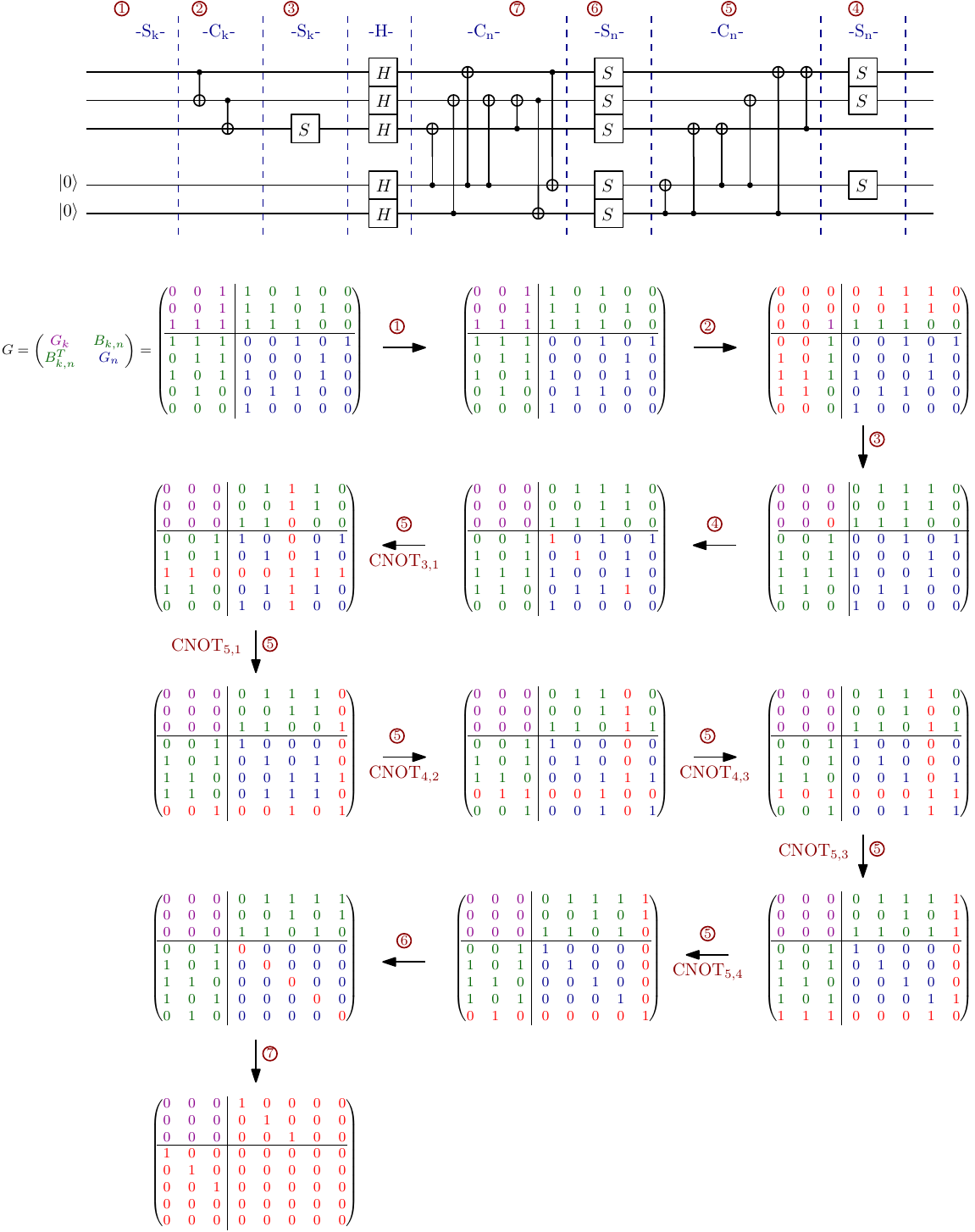}

\caption{Normal form S\textsubscript{k}-C\textsubscript{k}-S\textsubscript{k}-H-C\textsubscript{n}-S\textsubscript{n}-C\textsubscript{n}-S\textsubscript{n}-H\textsubscript{n} for the Clifford isometry given in Fig.~\ref{fig::graph}. We show how each layer acts on the original graph state matrix to eventually reduce it to the identity. To do so, we recall that we read the circuit on the input side from left to right and the circuit on the output side from right to left, as explained in Figures \ref{synthesis_circuit} and \ref{fig::graph}. Two -C- layers implement triangular operators. We explicit a step-by-step simulation of the right-most -C- layer to show the effect on both $B_{k,n}$ and $G_n$.}
\label{fig::normal_form2}
\end{figure}

\subsubsection{Other normal forms for Clifford operators.}

We show how other normal forms in the literature for Clifford operators can also be derived from our framework.

\paragraph{-C-CZ-S-H-S-CZ-C- \cite{maslov2018shorter}}

This normal form consists in an application of a stage of CNOT gates, then a stage of CZ gates and finally a stage of S stages on both the input and the output side (remember that the operations done on the output side are reversed). Our normal form also has an additional -H- layer at an ultimate layer.

It is straightforward to define suitable operators for each stage: 
\begin{itemize}
  \item the -C- stage on the input side can implement any operator $C_1$ and the -C- stage on the output side can implement any operator $C_2$, as long as $C_1 B C_2 = I$,
  \item the -CZ- and -S- stages on the input side are defined such that it implements the graph state $C_1 G_k C_1^T$,
  \item the -CZ- and -S- stages on the output side are defined such that it implements the graph state $C_2^T G_n C_2$.
\end{itemize}

Naturally, as explained in \cite{maslov2018shorter}, the order of the -C-, -CZ- and -S- stages are free. In fact, as long as a suitable phase polynomial is implemented, then we have a valid implementation of the Clifford operator. 

\paragraph{-C-CZ-S-H-CZ-S-H- \cite{bataille2021reduced}}

Similar to the normal form -C-CZ-S-H-S-CZ-C-, this normal form highlights the fact that only one -C- stage is sufficient on one of the two sides. Again, here is a viable construction in our framework: 
\begin{itemize}
  \item reduce $B$ to the identity with the -C- stage,
  \item reduce the new $G_k$ and $G_n$ to zero with the -CZ- and -S- stages on both input and output sides.
\end{itemize}

\subsubsection{Normal forms for stabilizer states}

We explicit some normal forms for the special case of stabilizer states. In the case $k=0$ the normal forms greatly simplify because we only need to reduce a symmetric matrix $G_n$ to $0$ with no input side as all input qubits are in the initial state $\ket{0}$.

\paragraph{-H-CZ\textsubscript{n}-S\textsubscript{n}-H\textsubscript{n}}

Stabilizer states are equivalent to graph states up to local Clifford, see e.g \cite{PhysRevA.69.022316} for a proof. Given the definition of a graph state the normal form follows. In our framework, the -CZ- and -S- layers zero $G_n$ component by component.

\paragraph{-H-S\textsubscript{n}-C\textsubscript{n}-S\textsubscript{n}-H\textsubscript{n}}

Using a similar technique: writing $G_n = LDL^T \oplus D'$, we zero $G_n$ by first adding $D'$ with a layer of $S$ gates, then the CNOT layer implements $L^{-1}$ and the second layer of $S$ implements $D$.

\subsection{Synthesis in \{CNOT, S, CZ\}}

Overall, the structure of the normal forms given in Section~\ref{normal_forms} are very similar: two phase polynomials are implemented from either side of the Hadamard layer, the differences being in the way those phase polynomials are implemented. If we restrict the synthesis of $G$ in the gate set \{CNOT, S, CZ\}, we necessarily implement phase polynomials and we show that only a partial specification of the phase polynomials is sufficient.

\begin{theorem}

Let $V$ be a $k$ to $n$ Clifford isometry $V$ under graph-state form specified by its graph matrix $G = \begin{bmatrix} G_k & B_{k,n} \\ B_{k,n}^T & G_n \end{bmatrix}$. An operator of the form
\begin{equation} P_1 H^{\otimes n} P_2 \label{pp_eq} \end{equation}
where $P_1, P_2$ are phase polynomials implements $V$ if and only if 
\[ P_1^{\dag} = \sum_{x \in \{0,1\}^n} e^{i\frac{\pi}{2} x^TG_nx} \ket{A_1x}\bra{x} \]
\[ P_2 = \sum_{x \in \{0,1\}^k} e^{i\frac{\pi}{2} x^TG_kx} \ket{A_2x}\ket{0}^{\otimes n-k}\bra{x} \]
such that
\[ \begin{bmatrix} A_2^T & 0 \end{bmatrix} A_1 = B_{k,n} \]
where $A_1$ is $n \times n$, $A_2$ is $k \times k$.

\end{theorem}

\begin{proof}

$P_1 H^{\otimes n} P_2$ is a valid implementation of $V$ if and only if 
\[ H^{\otimes n}P_2 = P_1^{\dag} V. \]

We write 
\[ P_1^{\dag} = \left(\sum_{x \in \{0,1\}^n} e^{i\frac{\pi}{2} x^T\Gamma_1x} \ket{A_1x}\bra{x} \right),  \]
\[ P_2 = \left(\sum_{x \in \{0,1\}^k} e^{i\frac{\pi}{2} x^T\Gamma_2x} \ket{A_2x}\ket{0}^{\otimes n-k}\bra{x} \right),  \]
and we expand the two terms: 
\begin{align*} P_1^{\dag} V  & = \left(\sum_{y \in \{0,1\}^n} e^{i\frac{\pi}{2} y^T\Gamma_1y} \ket{A_1y}\bra{y} \right) \left( \frac{1}{\sqrt{2}^n}\sum_{x \in \{0,1\}^k} \sum_{y \in \{0,1\}^n} e^{i \frac{\pi}{2} \left(x^T G_k x + y^TG_ny \right) + i\pi x^T B_{k,n} y} \ket{y} \bra{x} \right)  \\
& = \frac{1}{\sqrt{2}^n}\left(  \sum_{x \in \{0,1\}^k} \sum_{y \in \{0,1\}^n} e^{i \frac{\pi}{2} \left(x^T G_k x + y^T(G_n \oplus \Gamma_1)y \right) + i\pi x^T B_{k,n} y} \ket{A_1y}\bra{x} \right) \\
& = \frac{1}{\sqrt{2}^n} \left( \sum_{x \in \{0,1\}^k} \sum_{y \in \{0,1\}^n} e^{i \frac{\pi}{2} \left(x^T G_k x + y^TA_1^{-1}(G_n \oplus \Gamma_1)A_1^{-1}y \right) + i\pi x^T B_{k,n} A_1^{-1}y} \ket{y}\bra{x}\right) 
\end{align*}

\begin{align*} 
H^{\otimes n}P_2 & = \frac{1}{\sqrt{2}^n} \left( \sum_{x \in \{0,1\}^n} \sum_{y \in \{0,1\}^n} e^{i\pi x^Ty} \ket{y}\bra{x} \right)\left(\sum_{x \in \{0,1\}^k} e^{i\frac{\pi}{2} x^T\Gamma_2x} \ket{A_2x}\ket{0}^{\otimes n-k}\bra{x} \right) \\
& = \frac{1}{\sqrt{2}^n}  \left( \sum_{x \in \{0,1\}^k} \sum_{x' \in \{0,1\}^{n-k}} \sum_{y \in \{0,1\}^n} e^{i\pi [A_2x ; x']^Ty} \ket{y}\bra{A_2x}\bra{x'} \right)\left(\sum_{x \in \{0,1\}^k} e^{i\frac{\pi}{2} x^T\Gamma_2x} \ket{A_2x}\ket{0}^{\otimes n-k}\bra{x} \right) \\
& = \frac{1}{\sqrt{2}^n}  \left( \sum_{x \in \{0,1\}^k} \sum_{y \in \{0,1\}^n} e^{i\pi x^TA_2^TE_ky} e^{i\frac{\pi}{2} x^T\Gamma_2x} \ket{y}\bra{x}  \right) 
\end{align*}

We have the equality if and only if 
\[ \forall x \in \{0,1\}^k, \forall y \in \{0,1\}^n, e^{i \frac{\pi}{2} \left(x^T G_k x + y^TA_1^{-1}(G_n \oplus \Gamma_1)A_1^{-1}y \right) + i\pi x^T B_{k,n} A_1^{-1}y} = e^{i\pi x^TA_2^TE_ky} e^{i\frac{\pi}{2} x^T\Gamma_2x} \]
i.e
\[ G_k = \Gamma_2, \]
\[ G_n = \Gamma_1, \]
and
\[ \begin{bmatrix} A_2^T & 0 \end{bmatrix} A_1 = B_{k,n}. \]

\end{proof}

\begin{corollary}
Let $\ket{G}$ be a graph state. Then we have 
\[ \ket{G} = \left( \sum_x e^{i\frac{\pi}{2} \cdot x^TGx} \ket{Ax}\bra{x} \right)^{\dag} H^{\otimes n} \ket{0}^{\otimes n} \]
for any invertible boolean matrix $A$.
In other words, the synthesis of a graph state in the gate set $\{CNOT, CZ, S\}$ is equivalent to the synthesis of a Clifford phase of a phase polynomial given by the graph matrix $G$ without any restriction on the final linear reversible operator applied.
\end{corollary}

\subsection{Synthesis for an LNN architecture}

\begin{theorem}

Any $k$ to $n$ Clifford isometry can be executed in two-qubit depth $4n + 3k - 2$ on an LNN architecture.

\end{theorem}

\begin{proof}
Our proof mainly relies on the synthesis of CNOT circuits in depth $5n$ from \cite{kutin2007computation} and CZ circuits in depth $2n-2$ from \cite{maslov2018shorter}. 
First, we can reduce $G$ to 
\[ G = \begin{bmatrix} G'_k & B'^T \\ B' & 0 \end{bmatrix} \]
with a CZ circuit of depth $2n-2$ using the technique from \cite{maslov2018shorter}. Initially the circuit in \cite{maslov2018shorter} is of depth $2n+2$ but the last four layers are purely CNOT gates which do not affect the zeroing of $G_n$, therefore we are free to remove them and keep a circuit of depth $2n-2$. Then we need to take a look at how the algorithm from \cite{kutin2007computation} works. Given a linear reversible operator $A$ on $n$ qubits they apply two circuits: 
\begin{itemize}
    \item the first one of depth $2n$, applying an operator $C_1$ such that $A' = C_1A$ is northwest triangular. For instance on $5$ qubits, 
    \[ A' = C_1 A = \begin{bmatrix} * & * & * & * & 1 \\ * & * & * & 1 & \\ * & * & 1 && \\ * & 1 &&& \\ 1 &&&& \end{bmatrix} \]
    where $*$ stands for any boolean value, and the entries non specified are $0$.
    \item the second circuit, of depth $3n$, implements an operator $C_2$ that synthesizes the triangular operator: 
    \[ C_2A' = I.  \]
\end{itemize}
Overall, we get $A = C_1^{-1}C_2^{-1}$ as a CNOT circuit of depth $5n$. For the rest of the proof, we need to go into the details of the generation of $C_2$. An example on $5$ qubits is given in Fig.~\ref{kutin_process}. The general idea is to give a label $a_i = n-i+1, i= 1 \hdots n$ to each row of $A'$ and to sort those labels with a LNN compliant sorting network as shown in Fig.~\ref{sorting}. Each box will swap the rows (and the labels) and may perform a row operation between the two swapped rows. Namely if a row with label $a_i$ is swapped with a row of label $a_j$, $i<j$, the row with label $a_j$ may have a nonzero $i$-th component and if so we perform the operation 
\[ a_j \leftarrow a_j \oplus a_i  \]
to zero that component. Given the initial label positions, a label $a_j$ will have to "encounter" all labels $a_i$, $i < j$ during the sorting process and we are ensured that the j-th row will eventually be the j-th canonical vector.

Each box will therefore consist of either a SWAP operation or an operation of the form: 
\[ (u,v) \leftarrow (u\oplus v, u) \]
or 
\[ (u,v) \leftarrow (v, u\oplus v). \]

The SWAP operation requires 3 CNOT gates while the SWAP+CNOT requires in fact only 2 CNOT gates \cite{kutin2007computation}. Therefore each box will be of depth at most $3$ and the total circuit $C_2$, of depth $n$ as a box-based circuit, will be of depth $3n$ as a CNOT-based circuit. Note also that, in any case, we always have at some point during the execution of a box on two rows the sum of the values of the two rows that appear on one row of $A'$. This will be an important property in our proof.

\begin{figure}
\subfloat[The sorting network applied to the operator, as a box-based circuit. \label{sorting}]{\includegraphics[scale=0.56]{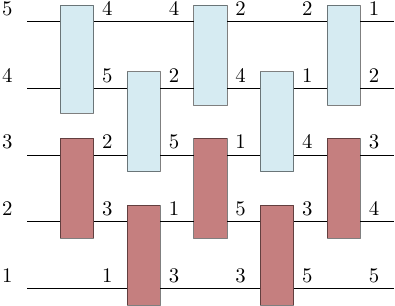}} \hspace{1cm}
\subfloat[The action of the sorting network on the operator for each layer.]{\includegraphics[scale=0.56]{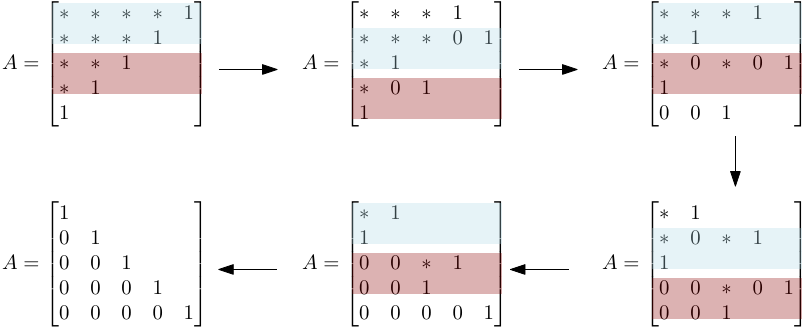}}
\caption{Illustration of the second part of the algorithm from \cite{kutin2007computation} to reduce a northwest triangular linear reversible operator to the identity.}
\label{kutin_process}
\end{figure}

Coming back to our proof, with the first half of the construction from \cite{kutin2007computation} we are able to reduce the CNOT part $B'$ to a northwest triangular matrix $B''$ in depth $2n$ as well. At this point we have an operator of the form 
\[ G = \begin{bmatrix} G''_k & \begin{pmatrix} B''^T &  0 \end{pmatrix} \\ \begin{pmatrix} B'' \\ 0 \end{pmatrix} & 0 \end{bmatrix} \]
where $B''$ is $k \times k$ and northwest triangular.

We now show that we can reduce $B''$ to the identity and $G_k$ to zero in depth $3k$. In fact, we will show that we can solely insert suitable $R_x(\pi/2)$ gates in the CNOT circuit reducing $B''$ to the identity matrix (taken from \cite{kutin2007computation}) to reduce $G_k$ to $0$ as well.

With $G$ in the form $\begin{bmatrix} G_k & D^T \\ D & 0 \end{bmatrix}$ with arbitrary $D$, the action of an $R_x(\pi/2)$ gate on the output side is equivalent to a rank-one modification of $G_k$ by a vector $u$ given by one of the row of $D$. In other words, if we apply an $R_x(\pi/2)$ on qubit $j$, then 
\[ G_k \leftarrow G_k \oplus D[j,:]^TD[j,:]. \]

Therefore in our case we can reduce $G_k$ to 0 if the set 
\[ S = \{ u^Tu | {u \text{ appearing in $B''$ during the synthesis }} \} \]
 is a basis of the set of symmetric boolean matrices. 

We prove by induction that $S$ generates the set 
\[ P_n = \{ e_ie_i^T \; | \; i \in \llbracket 1,n \rrbracket \} \cup \{ e_ie_j^T  \oplus e_je_i^T \; | \; i,j \in \llbracket 1,n \rrbracket, i \neq j \} \] 
and therefore generates the set of symmetric boolean matrices.

This is trivial for $n=1$. 
Now suppose this is true for some $n$. We introduce similar notations to the ones in \cite{kutin2007computation}: 
\begin{itemize}
  \item each line $i$ of $B''$ has a value $c_i$, i.e a boolean vector of size $n+1$,
  \item each line $i$ of $B''$ carries a label $a_i \in \llbracket 1, n+1 \rrbracket$, independent of $c_i$.
\end{itemize}

For simplicity, at any moment during the algorithm from \cite{kutin2007computation}, we now refer as parity $k$ the line of $B''$ that carries the label $a_k$, and by extension the value of the line carrying the label $a_k$. The key is to see the algorithm from \cite{kutin2007computation} as moving parities in $B''$ and during the SWAP of the parities we add or not some parities to the others. For simplicity we also consider the reverse circuit, i.e, the one that transforms the identity operator $I$ to $B''$. This does not change the set $S$ but this will simplify the reasoning. We still refer as $B''$ the matrix at any point during the algorithm.

The first point is to notice that the $(n+1)$-th parity will never be added to the other parities. So we can conclude that the first $n$ parities, apart from their interactions with the $(n+1)$-th parity, will behave exactly as in the case with $n$ qubits, therefore they satisfy the property and generates the set $P_{n}$. \\

We only need to show that, with the new values obtained on the last parity, we can generate $P_{n+1}$, i.e, we need to generate the set 
$\{ e_{n+1}e_{n+1}^T \} \cup \{ e_{n+1}e_j^T \oplus e_je_{n+1}^T \; | \; j \in \llbracket 1,n \rrbracket \}$. \\ \ \\

Every value carried by the $(n+1)$-th parity is of the form 
\[ e_{n+1} \oplus u_k \]
for a family of boolean vectors $(u_k)_k$ of size $n+1$ with the $(n+1)$-th component always equal to $0$. Applying an $R_x$ in two occasions gives the transformation 
\[ A \leftarrow A \oplus (e_{n+1} \oplus u_i)(e_{n+1} \oplus u_i)^T \oplus (e_{n+1} \oplus u_j)(e_{n+1} \oplus u_j)^T \]
\[ A \leftarrow A \oplus u_iu_i^T \oplus u_ie_{n+1}^T \oplus e_{n+1}u_i^T \oplus u_ju_j^T \oplus u_je_{n+1}^T \oplus e_{n+1}u_j^T \]
\[ A \leftarrow A \oplus \underbrace{u_iu_i^T \oplus u_ju_j^T}_{\in P_n} \oplus (u_i\oplus u_j)e_{n+1}^T \oplus e_{n+1}(u_i \oplus u_j)^T  \]

This means that, up to terms in $P_n$ that can be removed with suitable $R_x$ gates on the first $n$ parities, we can generate the set $e_{n+1} (\oplus_k \alpha_k u_k)^T \oplus (\oplus_k \alpha_k u_k)e_{n+1}^T$ for any $\alpha_k \in \{0,1\}$, i.e, any linear combination of the $(u_k)$'s.
What remains is to show that the family of vectors $(u_k)_k$ is a basis such that we can generate all the $e_i$'s and this concludes the proof.

We show that any new $u_p$ appearing during the process is not a linear combination of the previous $(u_k)_{k=1 \hdots p-1}$, given that we have $n$ such $u_k$ we get a basis $(u_k)_{k=1 \hdots n}$. By construction of the algorithm the $(n+1)$-th parity is initially on the last row of $B''$ and is equal to $e_{n+1}$. During the synthesis this parity will be swapped, and sometimes added, with the parity "above".

In other words, if the $(n+1)$-th parity is currently on line $k$ of $B''$, we have 
\[ c_k = e_{n+1} \oplus v_{k-1} \]
where $v_{k-1}$ is necessarily a linear combination of the set of parities located lines $k+1\hdots n$, i.e, the lines "below". The family $(u_j)_{j=1 \hdots k-1}$ also consists in linear combinations of the set of parities located lines $k+1\hdots n$. During the swap/addition with the parity above we necessarily create the value 
\[ e_{n+1} \oplus v_{k-1} \oplus c_{k-1} \]
and we set $u_{k} = v_{k-1} \oplus c_{k-1}$. $u_k$ is necessarily linearly independent with the previous $(u_j)_{j=1 \hdots k-1}$ otherwise the lines $k+1$ to $n$ and the $k-1$-th line of $B''$ are not linearly independent which is not possible by invertibility of $B''$. 

We set $v_{k} = v_{k-1}$ or $v_k = v_{k-1} \oplus c_{k+1} $, depending on the value of the $(n+1)$-th parity of $B''$ and we pursue the algorithm.
\end{proof}

\begin{corollary}
Any n-qubit Clifford operator can be executed in two-qubit depth $7n-2$ on an LNN architecture.
\end{corollary}

A similar result was independently found in \cite{maslov2023cnot}. They first provided a circuit of depth $7n+2$ which was then improved to $7n-4$. Note that in our proof we showed it is possible to work exclusively in the output side by using \{CNOT+S\} (to zero $G_n$) then \{CNOT+Rx\} circuits (to reduce $B_{k,n}$ to the identity while zeroing $G_k$). We slightly deviated from the need to implement two phase polynomials on both sides which gives, in our opinion, good insights about the variety of techniques that can be used in Clifford circuits synthesis.

\section{Practical algorithms for graph state synthesis and application to the codiagonalization of Pauli rotations}
\label{section:algorithms}

We give two algorithms for the synthesis of graph states, a special case of Clifford isometry where $k=0$. One algorithm optimizes the number of two-qubit gates while the second one optimizes the depth of the circuit, both assuming an all-to-all connectivity between the physical qubits. Then we use our algorithms in quantum compilers that synthesize a sequence of Pauli rotations. Usually such sequences are implemented by packing the rotations into groups of commuting operators, each group is then codiagonalized with a Clifford circuit and implemented with the synthesis of a phase polynomial. An efficient graph-state synthesis algorithm will improve the codiagonalization part of this process. Then in Section~\ref{section::algo_iso} we will extend the algorithms for graph state synthesis to deal with the case of any arbitrary Clifford isometry.

\subsection{Practical algorithms: a variant of the syndrome decoding based algorithm and a depth-optimized greedy algorithm}

Following our framework, $G_k$ is a $0 \times 0$ matrix and $B_{k,n}$ a $0 \times n$ matrix. Thus, the synthesis only consists in reducing a symmetric $n \times n$ boolean matrix $G_n$ to the zero matrix with the following operations: 
\begin{itemize}
    \item flip entries $G_n[i,j]$ and $G_n[j,i]$,
    \item add row $i$ to row $j$ and column $i$ to column $j$, $i \neq j$,
    \item rank-one updates as described in Table~\ref{framework}.
\end{itemize}

We will restrict ourselves to the gate set $\{CNOT, CZ, S\}$ which means that the rank-one updates are not available. The algorithms are variants of two algorithms designed by some of the authors for the synthesis of CNOT circuits \cite{DBLP:conf/rc/BrugiereBVMA20, de2021decoding, de2021reducing}. We present a sketch of how the algorithms work by highlighting the differences with CNOT circuits but more details can be found in the original articles. 

\subsubsection{The syndrome decoding based algorithm \cite{DBLP:conf/rc/BrugiereBVMA20, de2021decoding}.}

This algorithm aims at optimizing the total number of two-qubit gates in the circuit. The algorithm is recursive, suppose we managed to reduce $G_n$ to 
\[ G_n = \begin{bmatrix} 0_{n-1} & s^T \\ s & \star \end{bmatrix}  \]
with a sequence of operations $C$ acting on the first $n-1$ qubits where $0_{n-1}$ is the $n-1 \times n-1$ zero matrix, $s$ is an arbitary vector of size $n-1$ and $\star$ can be $0$ or $1$. We show how to efficiently add elementary operations to $C$ to create a new sequence $C'$ such that its execution will zero $s$ while keeping the $n-1 \times n-1$ submatrix equal to $0$. $G[n,n]$ can be ultimately zeroed with an S gate and does not participate in the increase of the 2-qubit gate count. Overall the algorithm consists in:
\begin{itemize} 
    \item a recursive call on a submatrix of $G_n$ of size $n-1$, giving a sequence $C$,
    \item and a call to the procedure that will treat separately the last line and column of $G_n$.
\end{itemize}

One key of the algorithm is to notice that at any time during the execution of $C$ we are free to do the following operations: 
\begin{itemize}
    \item flip any entry of $s$ with a CZ gate, by symmetry this will flip the same entry in $s^T$,
    \item add any row $1 \leq i \leq n-1$ of the current value of $G_n$ to $s$ with a CNOT gate, by symmetry the column operation will also modify $s^T$,
    \item flip any diagonal entry $i$ of $G[1:n-1,1:n-1]$ with an S gate, then add row $i$ to $s$ with a CNOT gate and flip again the $i-th$ diagonal entry.
\end{itemize}

These operations will only modify $s$ in $G_n$ and keep $G_n[1:n-1,1:n-1]$ unchanged and therefore ultimately $0$. All these operations have a two-qubit cost of $1$. Overall any of these three operations we can be seen as adding a vector to $s$. Therefore, the procedure to zero $s$ consists in two steps:
\begin{enumerate}
    \item scan the matrix $G[1:n-1,1:n-1]$ during its zeroing and store all the vectors that can be added to $s$ in a matrix $H$. $H$ is eventually $m \times (n-1)$ where $m$ is the number of different parities encountered,
    \item solve $xH = s$ with $x$ of smallest Hamming weight possible. Each vector to add costs one two-qubit gate so the smallest the Hamming weight of $x$, the smallest the number of extra two-qubit gates.
\end{enumerate}

Step 2 is also known as the syndrome decoding problem and we use a greedy method to solve it: we choose the vector $v$ in $H$ that minimizes the Hamming weight $v \oplus s$, we set $s \leftarrow s \oplus v$ and we repeat this until $s$ is $0$. The presence of the canonical vectors guarantees that the procedure eventually terminates. To improve the efficiency of the greedy method, it is possible to solve $xHP = sP$ for different invertible matrices $P$ and keep the best result. This method is known as the Information Set Decoding strategy and significant improvements can be obtained if sufficient iterations are done. The only constraint on $P$ is that the canonical vectors must be present in $HP$. 

Once the syndrome decoding problem is solved, we get a sequence of vectors, each corresponding to a specific operation (i.e, its type -- CNOT, CZ or S+CNOT+S -- and the index where it should be inserted inside $C$ to perform the desired operation). To create the new sequence of $C'$ we just need to update $C$ by inserting accordingly the operations.

\subsubsection{A depth-optimized greedy algorithm.} \label{sec::depth_opti_graph_state}

We use the following greedy algorithm to optimize the depth. Until $G$ is zero repeat the following procedure: 
\begin{enumerate}
    \item initialize an empty list $L$ of used qubits and an empty set of operations $O$. 
    \item At each iteration choose the row or column operations among the non-used qubits that reduces the Hamming weight of $G_n$ the most. If the decrease in the Hamming weight is strictly larger than $1$, then perform the operation, add it to $O$ and add the two qubits to $L$. 
    \item When no such operation is available, compute the set 
    \[ Q = \llbracket 1, n \rrbracket \setminus L, \] 
    and consider the submatrix $G_n[Q, Q]$: this is the submatrix on which entries can still be flipped without increasing the depth of $O$. To compute the largest number of entries that can be flipped in depth $1$ see $G_n[Q, Q]$ as the adjacency matrix of a graph. Each $1$ entry in $G[Q, Q]$ is an edge in the graph. Then a maximum matching in this graph gives a set of maximum size of entries that can be flipped. 
    \item One step of this procedure produces a set $O$ of depth $1$. Add $O$ to the full set of operations done on $G_n$.
\end{enumerate}    

We could have avoided the use of $O$ in the description of our algorithm but we wanted to highlight the principle of the algorithm: at each step we compute a set of operations that reduces the most the number of $1$ in $G_n$ and that can be executed in depth $1$. The algorithm always terminates as at each step we strictly reduce the number of nonzero elements in $G_n$.

\subsection{Application to the codiagonalization of Pauli rotations}

Given a set of $k$ Pauli operators on $n$-qubit $\{P_1, P_2, \hdots, P_k \}$, if they all commute then there exists a Clifford operator $C$ such that 
\[ CP_iC^{\dag} = Z^{a_i} = Z_1^{a_{i,1}} \otimes Z_2^{a_{i,2}} \otimes \hdots \otimes Z_n^{a_{i,n}} \forall i \in \llbracket 1,k \rrbracket, a_i \in \mathbb{F}_2^n. \]

$C$ is not unique and we can look for an optimized one using our graph synthesis framework. If we write our Pauli operators as a tableau
\[ \begin{bmatrix} Z \\ X \end{bmatrix} \in \mathbb{F}_2^{n \times k}, \]
the codiagonalization problem is then equivalent to finding a Clifford operator that will zero the $X$ part. We propose the following procedure: 
\begin{enumerate}
    \item considering the following block structure of our tableau 
    \[ \begin{bmatrix} Z_k \\ Z_{n-k} \\ X_k \\ X_{n-k} \end{bmatrix} \in \mathbb{F}_2^{n \times k} \]
    where $Z_k, X_k$ are $k \times k$ and $Z_{n-k}, X_{n-k}$ are $(n-k) \times k$. Find a circuit of Hadamard gates and SWAP gates such that $X_k$ is full rank. 
    \item find a Clifford operator that zeroes $X_{n-k}$.
    \item with free column operations reduce $X_k$ to the identity operator,
    \item use the graph-state synthesis framework to zero $Z_k$,
    \item use Hadamard gates on the first $k$ qubits to commute $X_k$ with the zero block
\end{enumerate}

At the end of the procedure, we eventually have $X=0$. The main complexity of this procedure lies in steps 2 and 4. The Hadamard layers do not increase either the 2-qubit gate count or the 2-qubit gate depth. The SWAP gates can be transfered at the end of the circuit and post-processed. Although the column operations during step 3 are free, they change the final set of diagonalized Pauli rotations. This does not change the fact that the resulting Clifford circuit will codiagonalize our initial set of Pauli rotations but, once such a Clifford is computed, we need to reverse the column operations to get back the actual diagonal Pauli operators. Alternatively, one can execute the Clifford operator again on the initial set of Pauli rotations.

We already have algorithms for step 4, we now give methods to do step 2 as well. Those methods are variants of the ones we already used for graph state synthesis so we will address them quickly. 

\subsubsection{2-qubit count optimization with the syndrome decoding problem.}

By assumption $X_k$ is full rank. With suitable column operations we can have $X_k = I$. We zero each row of $X_{n-k}$ one by one starting from row $1$ to row $n-k$. During the $i$-th step we compute the rows that appeared in the zeroing of $X_{n-k}[1:i-1]$ that we gather in a matrix $H$. We also add the canonical vectors given in $X_k$. We similarly solve the syndrome decoding problem $Hx = X_{n-k}[i,:]$ to have the smallest set of rows from $H$ one has to add to $X_{n-k}[i,:]$ to zero the row. Those row operations can be performed with suitable CNOT gates inserted in the current Clifford circuit. Note that if we apply an Hadamard gate, resp. an $R_x(\pi/2)$ gate, on the $k+i$-th qubit, the row to zero becomes $Z_{n-k}[i,:]$, resp. $X_{n-k}[i,:] \oplus Z_{n-k}[i,:]$: we can solve the syndrome decoding problem with these three possibilities and keep the best result at the cost of adding one-qubit gates.

\subsubsection{Depth optimization with a greedy procedure.} This case is even closer to the CNOT case from \cite{de2021reducing} than the graph-state case. Consider the matrix 
\[ B = X_{n-k}X_k^{-1}. \]
Any row operation on $X_{n-k}$ is equivalent to a row operation on $B$, any row operation on $X_k$ is equivalent to a column operation on $B$ and adding a row of $X_k$ to a row of $X_{n-k}$ is equivalent to shifting an entry of $B$. Until $B$ is zero repeat the following procedure: 
\begin{enumerate}
    \item initialize two sets of non-used qubits $L_1 = \llbracket 1,k \rrbracket, L_2 = \llbracket k+1,k+n \rrbracket$,
    \item compute the row or column operation that reduces at most the number of $1$ in $B$ if such row or column operation involves only qubits in $L_1$ or $L_2$, perform this operation and remove the two used qubits from either $L_1$ or $L_2$,
    \item when no more row or column operation can be done consider the submatrix $B[L_2, L_1]$, this is the adjacency matrix of a bipartite graph $G$. Compute a maximum matching in $G$ that will give you the set of row operations to perform between $X_k$ and $X_{n-k}$.
\end{enumerate}

One iteration of the three steps above gives a CNOT circuit of depth $1$ that tries to reduce at most the number of one entries in $B$. The procedure always terminates as we can always remove the entry one by one during step 3 of the procedure above.

\subsubsection{Extra-optimization.}

Note that if the number of rotations $k$ is close to $n$, it might be preferable to complete the set of Pauli operators into $n$ operators that commute and directly apply the graph state synthesis on $n$ qubits. In the following benchmarks we will always try the two approaches and keep the best results. We also do the codiagonalization multiple times with random qubits permutation at the beginning of the algorithm, again we keep the best result and we post-processed the selected permutation.

\section{Practical algorithms for Clifford isometry synthesis} \label{section::algo_iso}

We reuse the methods used in Section~\ref{section:algorithms} to deal with the case of general Clifford isometries given by matrices of the form $\begin{pmatrix} G_k & B_k \\ B_k^T & G_n \end{pmatrix}$. Again we restrict ourselves to the gate set \{CNOT,CZ,S\}. Essentially, we show how to zero both $G_n$ and $G_k$ using the graph state synthesis techniques while reducing $B_k$ to the identity. 

\subsection{Count optimization with the syndrome decoding problem} \label{sec:count_cliff_opti}

For any $n \times k$ matrix $A$, there exists an $n \times n$ permutation matrix $P$, an $n \times k$ lower triangular operator $L$ and a $k \times k$ upper triangular matrix $U$ such that 
\[ A = PLU. \]

We use this so-called LU decomposition on $B_k^T$ to split our synthesis in three steps: 
\begin{enumerate}
    \item first, deal with the permutation matrix $P$ such that we get a matrix of the form 
    \[ \begin{pmatrix} G_k & (LU)^T \\ LU & G_n \end{pmatrix}, \] 
    \item then, reduce $G_n$ to $0$ while reducing $L$ to the identity with operations on the output side, this would give us an intermediate matrix of the form 
    \[ \begin{pmatrix} G_k & U^T \\ U & 0 \end{pmatrix}, \]
    \item finally, reduce $G_k$ to $0$ while reducing $U$ to the identity with operations on the input side to get the identity 
    \[ \begin{pmatrix} 0 & I_k^T \\ I_k & 0 \end{pmatrix}. \] 
\end{enumerate}

Once this is done we have a valid implementation of the isometry. 

\paragraph{Dealing with the permutation.} To deal with $P$, one can post-process it to the end of the circuit. If for some reason the permutation cannot be postponed, one can find a CNOT circuit on the output side such that the resulting $B_k$ has an LU decomposition with $P=I$. This is possible if and only if all its leading principal minors are nonzero (i.e, $B_k[:i,:i]$ is invertible for every $i=1...k$). Therefore, for all $i=1...k$, assuming all $B_k[:j,:j]$ are invertible for $j < i$, one can look for a CNOT gate between qubit $i$ and a qubit $i_2  > i$ to make $B_k[:i,:i]$ invertible. This is always possible as $B_k$ is of full rank. This procedure only requires at most $k$ additional CNOT gates.

\paragraph{Graph state + linear reversible triangular operator synthesis.} The last two steps essentially reduce to the same problem. In the case of step 2 we are given a $n \times (n+k)$ matrix of the form 
\[ \begin{pmatrix} L & G_n \end{pmatrix} \]

where $L$ is lower triangular and $G_n$ symmetric. The goal is to reduce it to \[ \begin{pmatrix} I_k & 0 \end{pmatrix} \]

with the following available operations: 
\begin{itemize}
    \item CZ and S gates will flip any entry of $G_n$,
    \item CNOT gates will act as row and column operations on $G_n$ while only acting as row operations on $L$.
\end{itemize}

Step 3 is dealt in a similar way on $k \times 2k$ matrices of the form 
\[ \begin{pmatrix} G_k & L \end{pmatrix} \]

where $L$ is $k \times k$.

We only show how to deal with step 2. A recursive formulation similar to the one presented in Section~\ref{section:algorithms} works. If we managed to deal with the first $n-1$ lines we end up with a matrix of the form 
\[ \begin{pmatrix} I_{k-1} & 0_{(n-1)} \\ s' & s \end{pmatrix} \]

where $s'$ is of size $\min(k,n)$ and $s$ is of size $n$. To complete the synthesis, we need to zero $s[1:n-1]$ and $s'$ (unless $k=n$, in this case $s'$ needs to be reduced to the canonical vector $e_n$ therefore only $s'[1:n-1]$ needs to be zeroed). 
We can just stack $s$ and $s'$ and get a larger syndrome decoding problem instance $xH = [s' s]$ to solve where the vectors in $H$ are also concatenation of rows in $L$ and rows in $G_n$. We are guaranteed that the $[s' s]$ can be zeroed by a greedy process: 
\begin{itemize}
    \item the canonical vectors for the $s$ part are available through the CZ gates,
    \item the final matrix form $\begin{pmatrix} I_{k-1} & 0_{(n-1)} \\ s' & s \end{pmatrix} $ that we get after having executed all the elementary operations on the first $n-1$ lines ensure that the canonical vectors for the $s'$ part are accessible through CNOT gates.
\end{itemize}

Again, we can improve the performance of the syndrome decoding solver by solving it with different changes of basis.

\subsection{Depth optimization}

The algorithm for depth optimization proceeds in a similar fashion that the count optimization one. Again, we apply an $LU$ decomposition to $B_k^T$ and we follow the same three steps as explained in Section~\ref{sec:count_cliff_opti}.

\paragraph{Dealing with the permutation.} It is known that any permutation can be implemented in depth $d=6$ in an all-to-all architecture \cite{DBLP:journals/siamcomp/MooreN01}. 

\paragraph{Graph state + linear reversible triangular operator synthesis.} To reduce a $n \times (n+k)$ matrix of the form 
\[ \begin{pmatrix} L & G_n \end{pmatrix} \]

with $L$ lower triangular to \[ \begin{pmatrix} I_k & 0 \end{pmatrix} \]

we follow the same greedy procedure as in Section~\ref{sec::depth_opti_graph_state} for regular graph state synthesis. The only difference is that we have to keep the triangular structure of $L$ in order to reduce it to the identity matrix. Essentially this reduces to only selecting CNOT gates between qubits such that row operations $L[i,:] \leftarrow L[i,:] \oplus L[j,:]$ with $i > j$ are performed if $i \leq k$. Note that there are no constraints on the rows $k+1 \hdots n$ of $L$ as this part is not triangular.

We are ensured that the greedy process can terminate because $CZ$ and $S$ gates can always be used to strictly decrease the Hamming weight of $G_n$, then CNOT gates will eventually reduce $L$ to $I_k$ because of the triangular structure of $L$. The first column of $L$ can always be reduced to $0$ using the first line of $L$, the second column with the second row and so forth. 

The reduction of the $k \times 2k$ matrix
\[ \begin{pmatrix} G_k & L \end{pmatrix} \]

is done in a similar way. Here the constraints apply to all the rows of the $k \times k$ matrix $L$.

\section{Benchmarks}
We implemented the different algorithms of Section \ref{section::algo_iso} in a python-binded Rust crate called \emph{Rustiq}\footnote{\href{https://github.com/smartiel/rustiq}{https://github.com/smartiel/rustiq}}. The library offers a straightforward python interface to call the various synthesis algorithms (and more) via basic python types.
In order to evaluate the performances of the algorithms introduced in section \ref{section:algorithms}, we ran benchmarks over both random and structured instances. Results are compared with the co-diagonalization heuristics described in \cite{cowtan2020generic} and with depth and count oriented Clifford synthesis heuristics from \cite{bravyi2021clifford,Maslov_2022}.

\subsection{Co-diagonalization and stabilizer state synthesis}

\noindent{\bf Random instances:} In the first benchmark, we simply pick some Clifford Tableaux uniformily at random and co-diagonalize their Z outputs. In other words, given some Clifford operator $C$, we use our heuristics to co-diagonalize operators $\{CZ_iC^\dagger, \forall i \in [1,n] \}$. We then compare the entangling gate count and entangling depth of the output circuits. Results are reported in Figures \ref{fig:random_count}, \ref{fig:random_depth},
\ref{fig:random_count_fixed}, \ref{fig:random_depth_fixed}. For full stabilizer state synthesis, we can observe that the Syndrome algorithm asymptotically outperforms the state-of-the-art method in both entangling gate count and entangling depth. 
Interestingly, so does the Matching algorithm, while it only aims at limiting the circuit's depth without being conservative in the entangling gate count. 
When using our heuristics to co-diagonalize incomplete set of rotations (Figure \ref{fig:random_count_fixed} and \ref{fig:random_depth_fixed}), the Syndrome algorithm performs similarly to the state-of-the-art heuristic until some break-point where it stagnates and thus starts outperforming it. This behavior can be observed in both entangling count and depth. The Matching algorithm is outperformed by the other two in entangling count until it also starts stagnating and ends up outperforming
the state-of-the-art.

\medskip

\begin{figure}[h]
    \begin{adjustwidth}{-1.5cm}{}
    \begin{center}
        \begin{tabular}{cc}
            \includegraphics[scale=0.5]{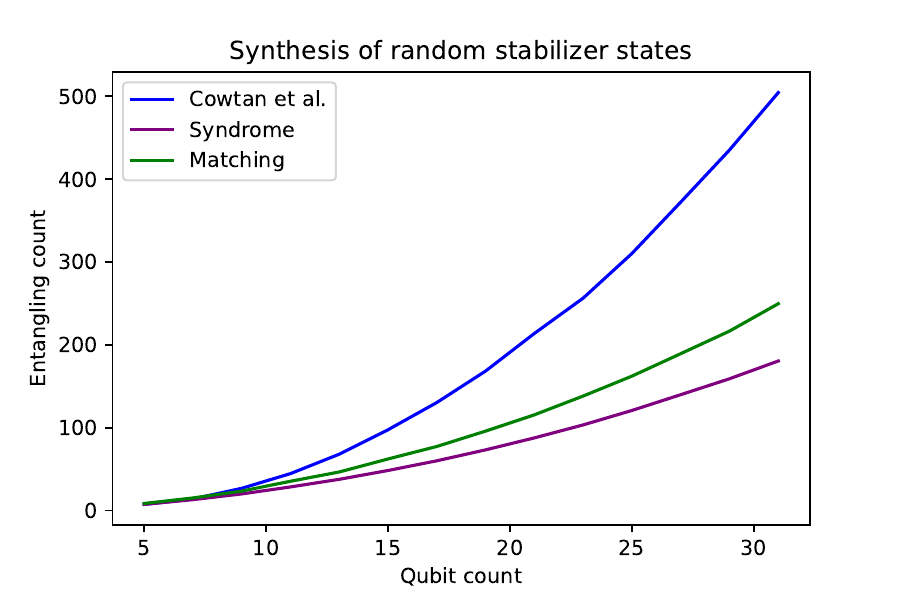}&\includegraphics[scale=0.5]{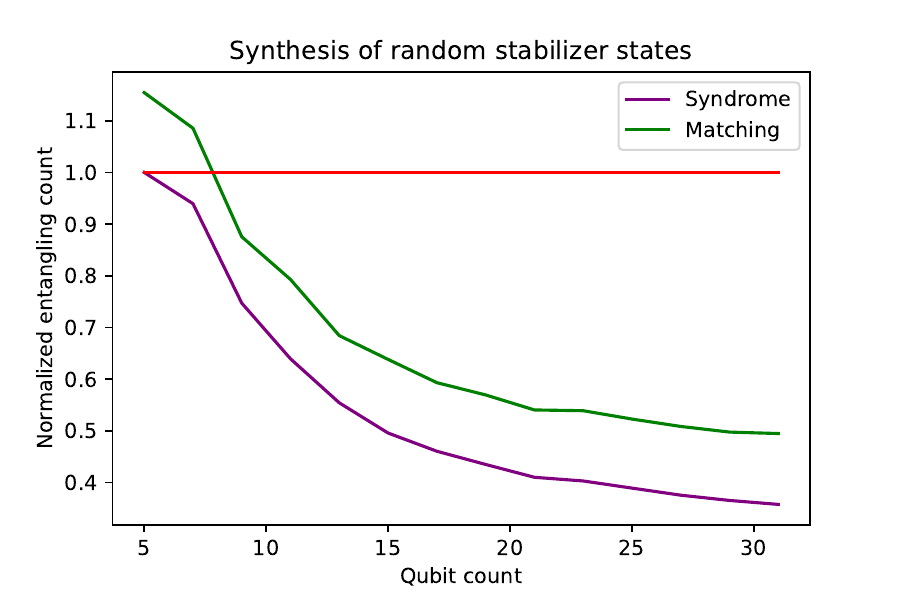}
        \end{tabular}
    \end{center}
    \end{adjustwidth}
    \caption{Left: Entangling count (as in CNOT equivalent count) as a function of the number of qubits. Each point is generated by averaging the entangling count of 40 random Clifford tableaux. Right: the same data, but normalized by the entangling count obtained via the state-of-the-art method of \cite{cowtan2020generic}.}\label{fig:random_count}
\end{figure}

\begin{figure}[h]
    \begin{adjustwidth}{-1.5cm}{}
    \begin{center}
        \begin{tabular}{cc}
            \includegraphics[scale=0.5]{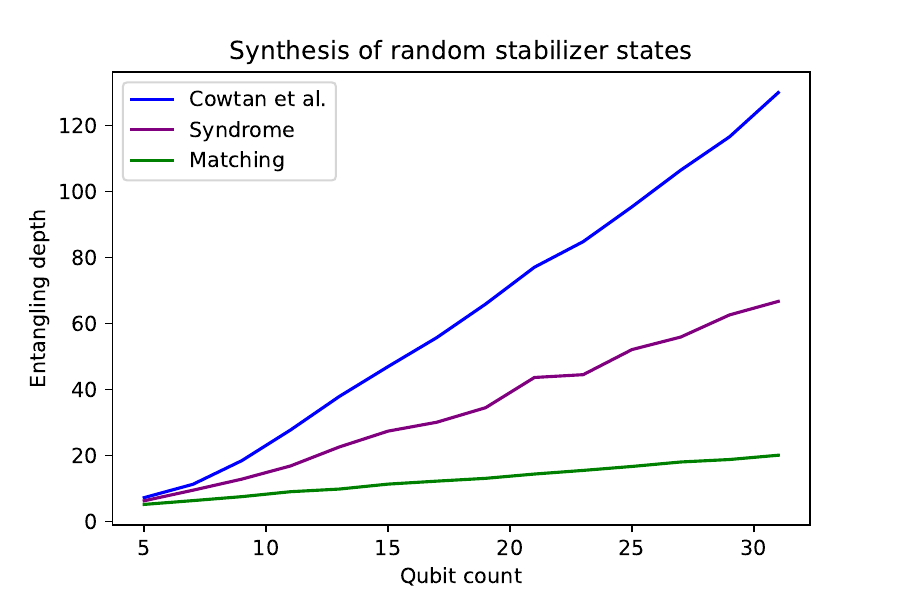}&\includegraphics[scale=0.5]{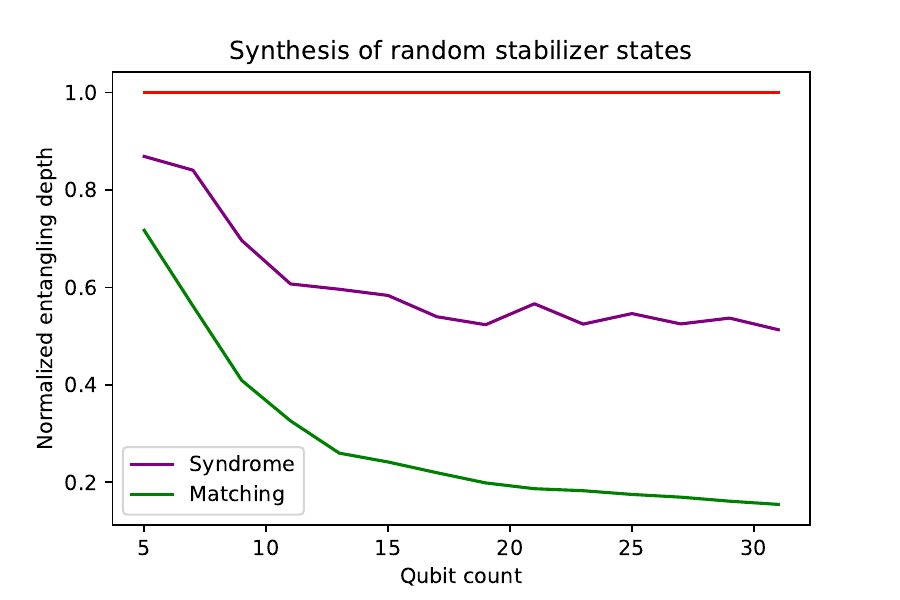}
        \end{tabular}
    \end{center}
    \end{adjustwidth}
    \caption{Left: Entangling depth (as in CNOT equivalent depth) as a function of the number of qubits. Each point is generated by averaging the entangling depth of 40 random Clifford tableaux. Right: the same data, but normalized by the entangling depth obtained via the state-of-the-art method of \cite{cowtan2020generic}.}\label{fig:random_depth}
\end{figure}

\begin{figure}[h]
    \begin{adjustwidth}{-1.5cm}{}
    \begin{center}
        \begin{tabular}{cc}
            \includegraphics[scale=0.5]{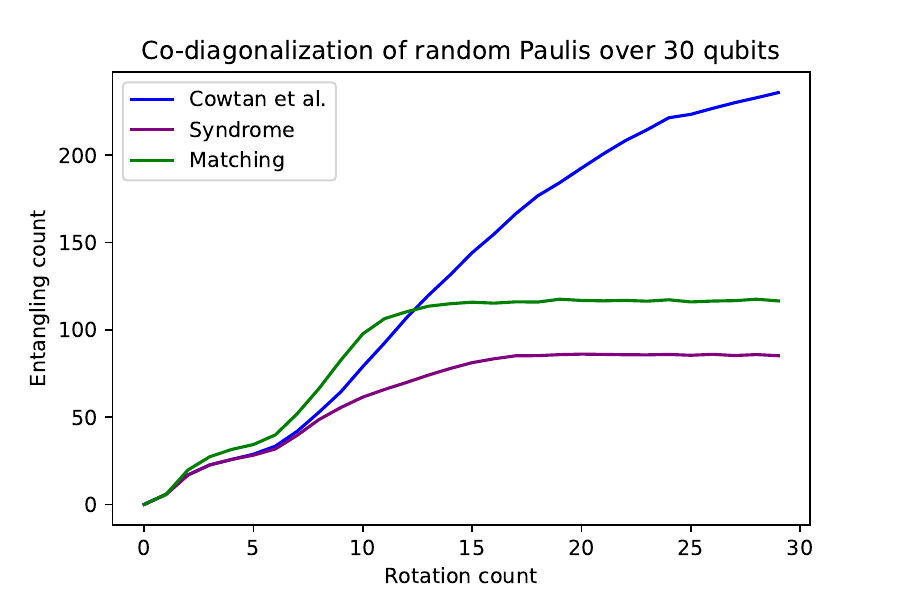}&\includegraphics[scale=0.5]{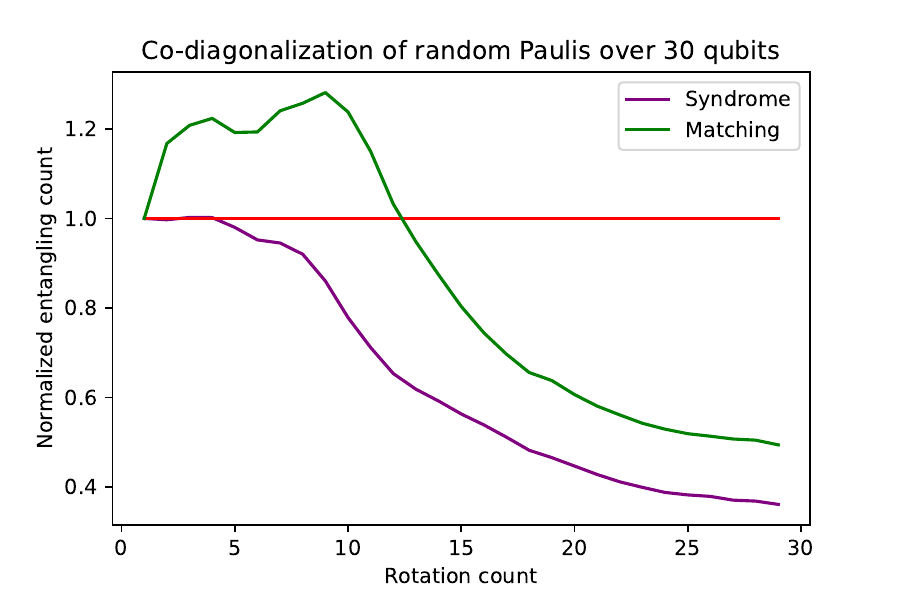}
        \end{tabular}
    \end{center}
    \end{adjustwidth}
    \caption{Left: Entangling count (as in CNOT equivalent count) as a function of the number of rotations in a commuting group of rotations over 30 qubits. Each point is generated by averaging the entangling count of 40 random rotation groups. Right: the same data, but normalized by the entangling count obtained via the state-of-the-art method of \cite{cowtan2020generic}.}\label{fig:random_count_fixed}
\end{figure}

\begin{figure}[h]
    \begin{adjustwidth}{-1.5cm}{}
    \begin{center}
        \begin{tabular}{cc}
            \includegraphics[scale=0.5]{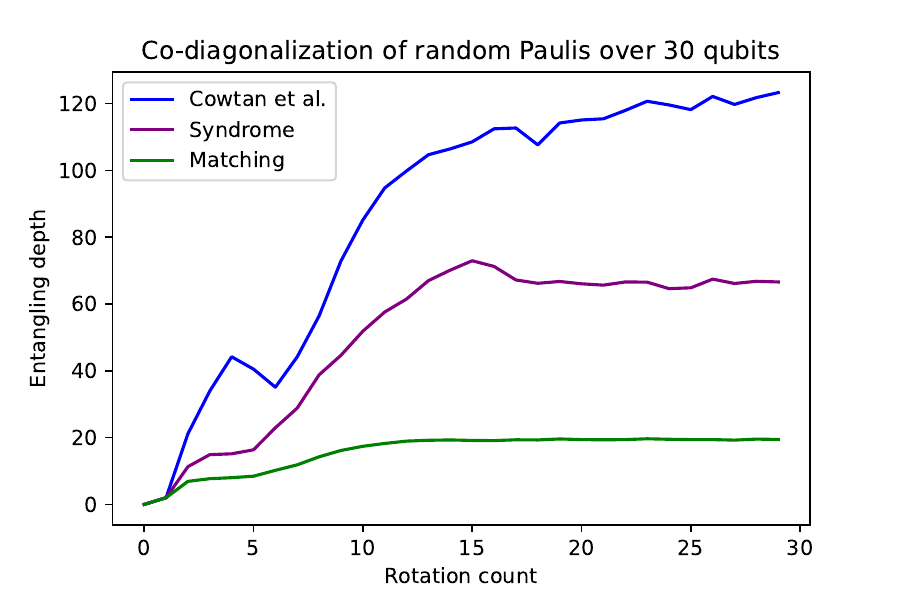}&\includegraphics[scale=0.5]{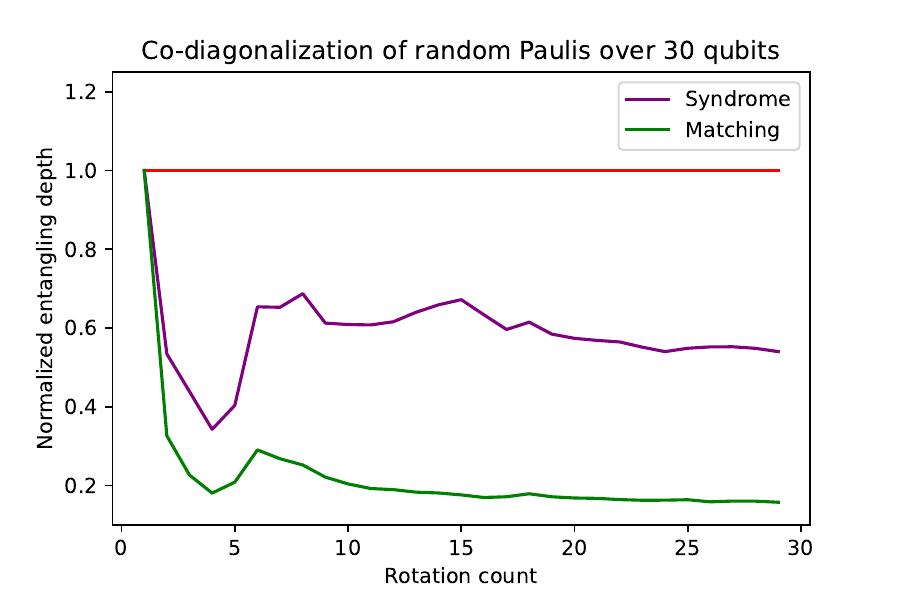}
        \end{tabular}
    \end{center}
    \end{adjustwidth}
    \caption{Left: Entangling depth (as in CNOT equivalent depth) as a function of the number of rotations in a commuting group of rotations over 30 qubits. Each point is generated by averaging the entangling depth of 40 random rotation groups. Right: the same data, but normalized by the entangling depth obtained via the state-of-the-art method of \cite{cowtan2020generic}.}\label{fig:random_depth_fixed}

\end{figure}

\noindent{\bf Structured instances:} In order to assess the performances of these algorithms over interesting use cases, we embedded them into the meta-heuristics presented in \cite{cowtan2020generic} that tackles synthesis of sequences of Pauli rotations (up to arbitrary commutation of these rotations). In this meta-heuristic, Pauli rotations are first  into buckets of pairwise commuting rotations. Then, rotations in each bucket are co-diagonalized (using any co-diagonalization routine), and
finally synthesized via a phase polynomial synthesis routine (GraySynth from \cite{amy2018controlled} in the formulation of \cite{cowtan2020generic}). Our two algorithms, together with the co-diagonalization routine from \cite{cowtan2020generic} effectively gives us three different algorithms to synthesize sequences of Pauli rotations. We ran these algorithm on the UCCSD (Unitary Coupled-Cluster Singles and Doubles) instances provided in \cite{cowtan2020generic}. 
Once again, we measured the entangling gate count and depth of the resulting co-diagonalization circuits produced during the synthesis, and reported them in Tables \ref{table:bench_mol_1},\ref{table:bench_mol_2},\ref{table:bench_mol_3}.

As expected, the Syndrome algorithm systematically outfperforms the Cowtan et al. approach. For large molecules, it reaches up to a $10\%$ CNOT count reduction. Interestingly, on these particularly structures instances, the Matching algorithm manages to significantly reduce the circuit depth (up to $71\%$ depth reduction) while only adding $1\%$ CNOT gates. This opens up the possibility of generating UCCSD Ansätze with a significatively lower CNOT depth than the one proposed in
\cite{cowtan2020generic}, as long as one can also improve the entangling depth of the phase polynomials in between the co-diagonalization circuits.

\begin{table}
\centering
\resizebox{1.1\columnwidth}{!}{
\begin{tabular}{ccccccccccccccccc}
     \toprule
     \toprule
    ~  &~  &~  &~  && \multicolumn{2}{c}{Cowtan et al.} && \multicolumn{4}{c}{Syndrome} && \multicolumn{4}{c}{Matching} \\
    \cmidrule(lr){6-7} \cmidrule(lr){9-12} \cmidrule(lr){14-17} 
     mol. & orbitals & enc. & basis  && count & depth && count & rel. count & depth & rel. depth && count & rel. count & depth & rel. depth \\
\\
H2 & cmplt & P & sto3g && 0 & 0 && 0&0\% & 0&0\% && 0&0\% & 0&0\%\\
H2 & cmplt & BK & sto3g && 0 & 0 && 0&0\% & 0&0\% && 0&0\% & 0&0\%\\
H2 & cmplt & JW & sto3g && 8 & 4 && 8&-0.00\% & 4&-0.00\% && 8&-0.00\% & 4&-0.00\%\\
H2 & cmplt & P & 631g && 50 & 46 && 46&-8.00\% & 38&-17.39\% && 50&-0.00\% & 38&-17.39\%\\
H2 & cmplt & BK & 631g && 50 & 48 && 50&-0.00\% & 48&-0.00\% && 52&4.00\% & 40&-16.67\%\\
H4 & cmplt & BK & sto3g && 66 & 38 && 64&-3.03\% & 40&5.26\% && 68&3.03\% & 40&5.26\%\\
H2 & cmplt & JW & 631g && 68 & 64 && 68&-0.00\% & 64&-0.00\% && 68&-0.00\% & 46&-28.12\%\\
H4 & cmplt & P & sto3g && 76 & 50 && 76&-0.00\% & 50&-0.00\% && 88&15.79\% & 62&24.00\%\\
LiH & frz & P & sto3g && 98 & 90 && 92&-6.12\% & 82&-8.89\% && 98&-0.00\% & 74&-17.78\%\\
H4 & cmplt & JW & sto3g && 98 & 50 && 98&-0.00\% & 48&-4.00\% && 98&-0.00\% & 36&-28.00\%\\
LiH & frz & JW & sto3g && 120 & 114 && 118&-1.67\% & 110&-3.51\% && 120&-0.00\% & 80&-29.82\%\\
LiH & frz & BK & sto3g && 134 & 122 && 128&-4.48\% & 106&-13.11\% && 142&5.97\% & 90&-26.23\%\\
NH & frz & JW & sto3g && 276 & 176 && 276&-0.00\% & 160&-9.09\% && 276&-0.00\% & 112&-36.36\%\\
NH & frz & P & sto3g && 334 & 242 && 330&-1.20\% & 206&-14.88\% && 368&10.18\% & 178&-26.45\%\\
H2O & frz & JW & sto3g && 428 & 236 && 428&-0.00\% & 218&-7.63\% && 428&-0.00\% & 164&-30.51\%\\
NH & cmplt & JW & sto3g && 428 & 236 && 428&-0.00\% & 218&-7.63\% && 428&-0.00\% & 164&-30.51\%\\
LiH & cmplt & JW & sto3g && 440 & 220 && 440&-0.00\% & 208&-5.45\% && 440&-0.00\% & 160&-27.27\%\\
LiH & cmplt & P & sto3g && 480 & 302 && 470&-2.08\% & 266&-11.92\% && 532&10.83\% & 248&-17.88\%\\
NH & frz & BK & sto3g && 486 & 426 && 452&-7.00\% & 274&-35.68\% && 524&7.82\% & 196&-53.99\%\\
CH2 & frz & JW & sto3g && 496 & 232 && 496&-0.00\% & 206&-11.21\% && 516&4.03\% & 136&-41.38\%\\
H2 & cmplt & P & ccpvdz && 520 & 510 && 484&-6.92\% & 446&-12.55\% && 522&0.38\% & 374&-26.67\%\\
LiH & frz & P & 631g && 520 & 510 && 496&-4.62\% & 464&-9.02\% && 522&0.38\% & 374&-26.67\%\\
H2O & frz & P & sto3g && 526 & 368 && 516&-1.90\% & 312&-15.22\% && 582&10.65\% & 274&-25.54\%\\
NH & cmplt & P & sto3g && 526 & 368 && 512&-2.66\% & 298&-19.02\% && 582&10.65\% & 278&-24.46\%\\
LiH & frz & JW & 631g && 554 & 528 && 554&-0.00\% & 526&-0.38\% && 556&0.36\% & 376&-28.79\%\\
H2 & cmplt & JW & ccpvdz && 554 & 528 && 554&-0.00\% & 528&-0.00\% && 556&0.36\% & 376&-28.79\%\\
     \bottomrule
 \end{tabular}}
\caption{Entangling count and depth of the co-diagonalization circuits produced by our algorithms (Syndrome and Matching) and the state-of-the-art routine (Cowtan et al.). For the Syndrome and Matching algorithms, we additionally display the improvement (or degradation) relative to the state of the art.}\label{table:bench_mol_1}
\end{table}
\begin{table}

\resizebox{1.1\columnwidth}{!}{
\begin{tabular}{ccccccccccccccccc}
     \toprule
     \toprule
    ~  &~  &~  &~  && \multicolumn{2}{c}{Cowtan et al.} && \multicolumn{4}{c}{Syndrome} && \multicolumn{4}{c}{Matching} \\
    \cmidrule(lr){6-7} \cmidrule(lr){9-12} \cmidrule(lr){14-17} 
     mol. & orbitals & enc. & basis  && count & depth && count & rel. count & depth & rel. depth && count & rel. count & depth & rel. depth \\
\\
H2 & cmplt & BK & ccpvdz && 646 & 596 && 622&-3.72\% & 540&-9.40\% && 654&1.24\% & 398&-33.22\%\\
LiH & frz & BK & 631g && 646 & 596 && 620&-4.02\% & 534&-10.40\% && 656&1.55\% & 400&-32.89\%\\
H2O & cmplt & JW & sto3g && 776 & 490 && 774&-0.26\% & 452&-7.76\% && 790&1.80\% & 322&-34.29\%\\
CH2 & frz & P & sto3g && 778 & 512 && 766&-1.54\% & 432&-15.62\% && 870&11.83\% & 332&-35.16\%\\
NH & cmplt & BK & sto3g && 898 & 658 && 846&-5.79\% & 484&-26.44\% && 978&8.91\% & 380&-42.25\%\\
H2O & frz & BK & sto3g && 898 & 658 && 846&-5.79\% & 500&-24.01\% && 978&8.91\% & 378&-42.55\%\\
LiH & cmplt & BK & sto3g && 948 & 652 && 904&-4.64\% & 518&-20.55\% && 1074&13.29\% & 380&-41.72\%\\
CH2 & cmplt & JW & sto3g && 1010 & 550 && 1004&-0.59\% & 464&-15.64\% && 1082&7.13\% & 336&-38.91\%\\
H4 & cmplt & JW & 631g && 1024 & 524 && 1024&-0.00\% & 488&-6.87\% && 1024&-0.00\% & 376&-28.24\%\\
H2O & cmplt & P & sto3g && 1084 & 734 && 1030&-4.98\% & 594&-19.07\% && 1188&9.59\% & 496&-32.43\%\\
CH2 & frz & BK & sto3g && 1094 & 710 && 1040&-4.94\% & 540&-23.94\% && 1188&8.59\% & 362&-49.01\%\\
H4 & cmplt & P & 631g && 1134 & 690 && 1118&-1.41\% & 606&-12.17\% && 1228&8.29\% & 552&-20.00\%\\
H4 & cmplt & BK & 631g && 1334 & 848 && 1302&-2.40\% & 690&-18.63\% && 1566&17.39\% & 668&-21.23\%\\
H2O & cmplt & BK & sto3g && 1406 & 972 && 1278&-9.10\% & 736&-24.28\% && 1470&4.55\% & 502&-48.35\%\\
CH2 & cmplt & P & sto3g && 1412 & 900 && 1356&-3.97\% & 728&-19.11\% && 1558&10.34\% & 506&-43.78\%\\
H8 & cmplt & JW & sto3g && 1534 & 710 && 1518&-1.04\% & 574&-19.15\% && 1680&9.52\% & 420&-40.85\%\\
LiH & frz & P & ccpvdz && 1808 & 1798 && 1748&-3.32\% & 1690&-6.01\% && 1810&0.11\% & 1252&-30.37\%\\
H8 & cmplt & P & sto3g && 1858 & 1074 && 1812&-2.48\% & 856&-20.30\% && 2014&8.40\% & 572&-46.74\%\\
LiH & frz & JW & ccpvdz && 1866 & 1808 && 1866&-0.00\% & 1806&-0.11\% && 1868&0.11\% & 1256&-30.53\%\\
H8 & cmplt & BK & sto3g && 1894 & 1096 && 1812&-4.33\% & 908&-17.15\% && 2176&14.89\% & 616&-43.80\%\\
LiH & frz & BK & ccpvdz && 2036 & 1952 && 1982&-2.65\% & 1828&-6.35\% && 2064&1.38\% & 1326&-32.07\%\\
CH2 & cmplt & BK & sto3g && 2372 & 1486 && 2142&-9.70\% & 1154&-22.34\% && 2442&2.95\% & 664&-55.32\%\\
LiH & cmplt & JW & 631g && 2482 & 1410 && 2474&-0.32\% & 1270&-9.93\% && 2526&1.77\% & 958&-32.06\%\\
LiH & cmplt & P & 631g && 3682 & 2402 && 3572&-2.99\% & 1862&-22.48\% && 3986&8.26\% & 1468&-38.88\%\\
NH & frz & JW & 631g && 4208 & 2636 && 4156&-1.24\% & 1966&-25.42\% && 4598&9.27\% & 1368&-48.10\%\\
H2 & cmplt & P & ccpvtz && 4498 & 4488 && 4416&-1.82\% & 4328&-3.57\% && 4500&0.04\% & 3074&-31.51\%\\
     \bottomrule
 \end{tabular}}
    \caption{See Table \ref{table:bench_mol_1}.}
    \label{table:bench_mol_2}
\end{table}
\begin{table}

\resizebox{1.1\columnwidth}{!}{
\begin{tabular}{ccccccccccccccccc}
     \toprule
     \toprule
    ~  &~  &~  &~  && \multicolumn{2}{c}{Cowtan et al.} && \multicolumn{4}{c}{Syndrome} && \multicolumn{4}{c}{Matching} \\
    \cmidrule(lr){6-7} \cmidrule(lr){9-12} \cmidrule(lr){14-17} 
     mol. & orbitals & enc. & basis  && count & depth && count & rel. count & depth & rel. depth && count & rel. count & depth & rel. depth \\
\\
H2 & cmplt & JW & ccpvtz && 4586 & 4488 && 4586&-0.00\% & 4488&-0.00\% && 4588&0.04\% & 3076&-31.46\%\\
NH & frz & P & 631g && 4838 & 2894 && 4710&-2.65\% & 2216&-23.43\% && 5282&9.18\% & 1508&-47.89\%\\
H2 & cmplt & BK & ccpvtz && 5056 & 4808 && 4932&-2.45\% & 4602&-4.28\% && 5066&0.20\% & 3154&-34.40\%\\
LiH & cmplt & BK & 631g && 6600 & 4344 && 5954&-9.79\% & 2954&-32.00\% && 7170&8.64\% & 1810&-58.33\%\\
NH & frz & BK & 631g && 6946 & 4380 && 6524&-6.08\% & 3132&-28.49\% && 7646&10.08\% & 1832&-58.17\%\\
NH & cmplt & JW & 631g && 7014 & 3748 && 6904&-1.57\% & 2868&-23.48\% && 7880&12.35\% & 2004&-46.53\%\\
H2O & frz & JW & 631g && 7436 & 3684 && 7318&-1.59\% & 2928&-20.52\% && 8546&14.93\% & 2134&-42.07\%\\
NH & cmplt & P & 631g && 7522 & 4104 && 7212&-4.12\% & 3120&-23.98\% && 8128&8.06\% & 1884&-54.09\%\\
H2O & frz & P & 631g && 8700 & 4834 && 8408&-3.36\% & 3666&-24.16\% && 9584&10.16\% & 2236&-53.74\%\\
C2H4 & frz & JW & sto3g && 8854 & 3912 && 8564&-3.28\% & 3004&-23.21\% && 10472&18.27\% & 1990&-49.13\%\\
LiH & cmplt & JW & ccpvdz && 9684 & 5480 && 9504&-1.86\% & 4748&-13.36\% && 10156&4.87\% & 3748&-31.61\%\\
H4 & cmplt & JW & ccpvdz && 10162 & 5334 && 10148&-0.14\% & 4774&-10.50\% && 10766&5.94\% & 3980&-25.38\%\\
H2O & frz & BK & 631g && 10746 & 6002 && 10040&-6.57\% & 4184&-30.29\% && 12168&13.23\% & 2402&-59.98\%\\
NH & cmplt & BK & 631g && 11956 & 7418 && 10486&-12.30\% & 4762&-35.80\% && 12526&4.77\% & 2546&-65.68\%\\
H2O & cmplt & JW & 631g && 12398 & 5948 && 12052&-2.79\% & 4674&-21.42\% && 14400&16.15\% & 3082&-48.18\%\\
C2H4 & frz & P & sto3g && 12660 & 6820 && 11646&-8.01\% & 4620&-32.26\% && 13404&5.88\% & 2556&-62.52\%\\
H4 & cmplt & P & ccpvdz && 13594 & 8446 && 13158&-3.21\% & 6580&-22.09\% && 14702&8.15\% & 5198&-38.46\%\\
LiH & cmplt & P & ccpvdz && 13992 & 8630 && 13418&-4.10\% & 6434&-25.45\% && 15102&7.93\% & 4738&-45.10\%\\
H4 & cmplt & BK & ccpvdz && 17058 & 10284 && 16328&-4.28\% & 7802&-24.13\% && 19360&13.50\% & 6072&-40.96\%\\
H2O & cmplt & P & 631g && 18242 & 9642 && 16742&-8.22\% & 6584&-31.72\% && 19458&6.67\% & 3674&-61.90\%\\
C2H4 & cmplt & JW & sto3g && 18734 & 9008 && 17538&-6.38\% & 6734&-25.24\% && 21654&15.59\% & 4114&-54.33\%\\
LiH & cmplt & BK & ccpvdz && 18906 & 11342 && 17460&-7.65\% & 7942&-29.98\% && 21090&11.55\% & 5410&-52.30\%\\
C2H4 & frz & BK & sto3g && 20720 & 9718 && 17730&-14.43\% & 6904&-28.96\% && 21914&5.76\% & 3422&-64.79\%\\
C2H4 & cmplt & P & sto3g && 22402 & 11642 && 20142&-10.09\% & 7826&-32.78\% && 23576&5.24\% & 4082&-64.94\%\\
H2O & cmplt & BK & 631g && 28008 & 14542 && 23914&-14.62\% & 9718&-33.17\% && 29076&3.81\% & 4548&-68.73\%\\
C2H4 & cmplt & BK & sto3g && 41642 & 19966 && 34076&-18.17\% & 12622&-36.78\% && 42030&0.93\% & 5774&-71.08\%\\
     \bottomrule
 \end{tabular}}

    \caption{See Table \ref{table:bench_mol_1}.}
    \label{table:bench_mol_3}

\end{table}

\subsection{Clifford operator synthesis}

We also compared our Clifford isometry synthesis algorithms with state-of-art heuristics for Clifford operator synthesis. We ran seven different synthesis algorithms: 
\begin{itemize}
    \item our count-optimized synthesis algorithm, using different number of syndrome decoding iterations (1, 10, 100, 500), (key {\bf Syndrome $n$}, for $n$ iterations)
    \item our depth-optimized synthesis algorithm (key {\bf Depth}).
    \item the count-optimized heuristic from \cite{bravyi2021clifford} (key {\bf bshm}),
    \item and the depth-optimized algorithm from \cite{Maslov_2022} (key {\bf mz}).
\end{itemize}

We generated random Clifford operators over different qubit counts and averaged the entangling count and depth over 20 instances for each input size. The results are presented in Figure \ref{fig:rnd_clifford_}.

Our syndrome decoding based heuristic seems to perform quite well for small number of qubits. 
Calling the syndrome decoding greedy solver 10 times per decoding instance seems enough to beat the state-of-the-art methods up to 60 qubits. Increasing that number only improves performances. However a quick numerical extrapolation show that even using 500 calls to the solver is not enough to reach state-the-art performances for a few hundred qubits. This leads us to believe that for any fixed number of calls to the decoding solver, the resulting heuristic will eventually be worse than the state-of-the-art heuristic of \cite{bravyi2021clifford}. This behavior was already observed in \cite{DBLP:conf/rc/BrugiereBVMA20, de2021decoding} for linear operator synthesis.
Nevertheless, we think that our heuristic remains competitive for practical applications where one might be ready to spend a large amount of time/resources in order to optimize a single Clifford portion in a larger circuit.

Our {\bf Depth} method seems to outperform the algorithm of \cite{Maslov_2022} even asymptotically. Of course, the main interest of the method presented in \cite{Maslov_2022} lies in the derivation of an upper bound on the produced circuit's depth, which we fail to provide in our work.

\begin{figure}[h!]
    \centering
    \includegraphics[scale=0.45]{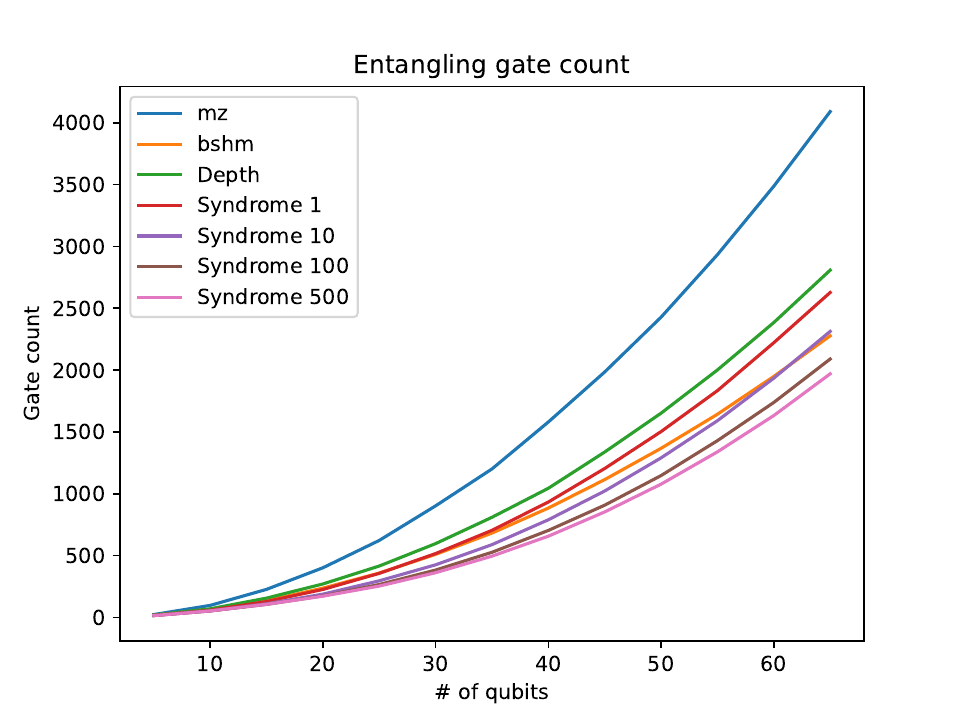}\includegraphics[scale=0.45]{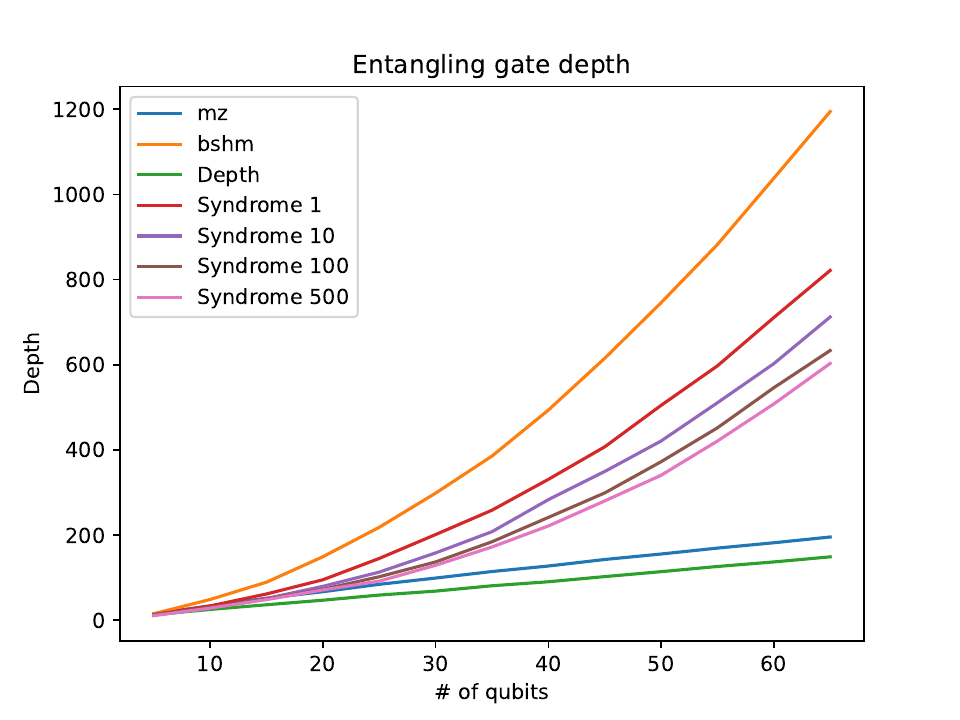}
    \caption{Entangling gate count and depth as a function of the number of qubits averaged over 20 random Clifford operators.}
    \label{fig:rnd_clifford_}
\end{figure}

\begin{figure}[h!]
    \centering
    \includegraphics[scale=0.45]{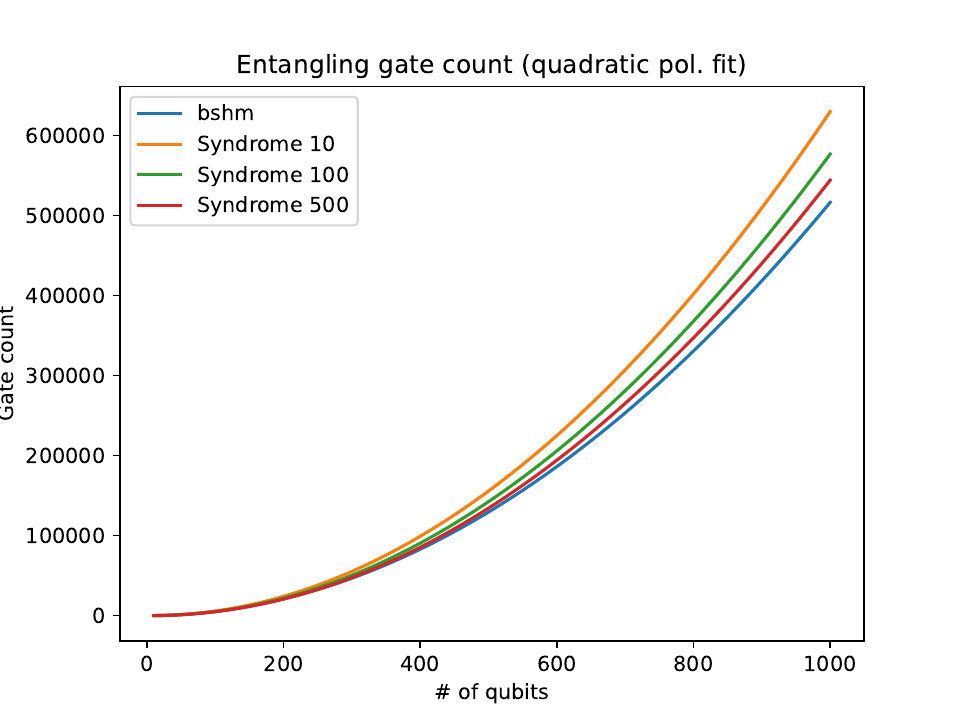}\includegraphics[scale=0.45]{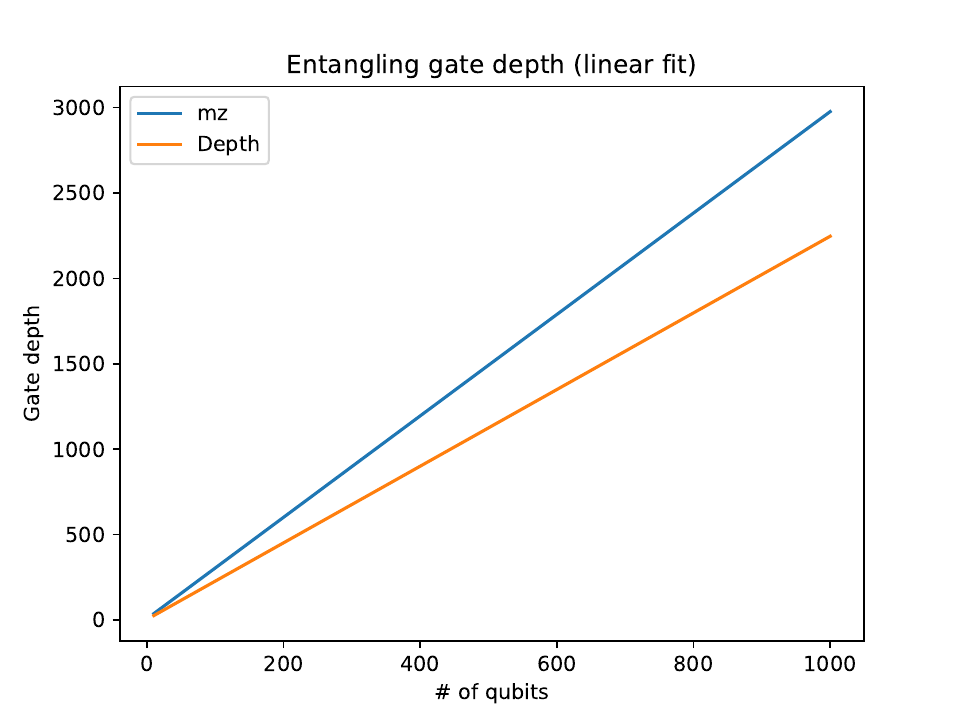}
    \caption{Extrapolated entangling gate count and depth for relevant methods. The fit are degree 2 polynomials for count and linear functions for depth. The fits were derived from the data presented in Figure \ref{fig:rnd_clifford_}.}
    \label{fig:rnd_clifford_ext}
\end{figure}

\section{Conclusion}

This article provides a generic compilation method for Clifford isometries. The framework exploits the graph-state structure of Clifford operators already exposed in previous works in the literature but with a minimal set of theoretical requirements. Overall, the synthesis of Clifford isometries is encoded as the reduction of a symmetric boolean matrix into a purposely defined identity matrix. The available operations are elementary: row and column operations, entries flip, and suitable rank-one updates. The versatility of our framework is highlighted by its capability of rediscovering almost every normal forms of the literature for Clifford operators and stabilizer states. Our framework also highlights deeper understandings about the role of phase polynomials and equivalent structures (namely CNOT + Rx circuits) in the construction of Clifford circuits. One application is the proof that any $n$-qubit Clifford circuit can be executed in two-qubit depth $7n$ on an LNN architecture. We also practically improve the codiagonalization of Pauli rotations and use it to optimize quantum circuits from quantum chemistry.

Still, while being simple and accessible, our framework lacks some universality. The Hadamard gate can only be applied in specific cases, and we are constrained to keep our Clifford isometries in a graph-state form. Therefore our framework can fail at offering the possibility to find optimal Clifford circuits for one given operator. As a future work it would be interesting to extend the graph-state formalism and see if a simple framework as universal as the tableau formalism and as simple as the graph-state approach can be derived.

\section*{Acknowledgments}

The authors thank Simon Perdrix for comments on an earlier version of this article. This work has been supported by the French state through the ANR as a part of \emph{Plan France 2030}, projects NISQ2LSQ (ANR-22-PETQ-0006) and EPiQ (ANR-22-PETQ-0007), as well as the ANR project SoftQPro (ANR-17-CE25-0009).

\bibliographystyle{quantum}
\bibliography{Biblio}

\section*{Action of the Clifford gates in the graph-state formalism}

\begin{enumerate} 
  \item $S$ on the output qubit $i$ ($1 \leq i \leq n$) \\
We perform the row operation $M[i,:] = M[i,:] \oplus M[i+n,:]$. We get 
\[ M = \left[ \begin{array}{c|c} G_nB \oplus \tilde{e}_{ii}B & (B^T)^{-1} \begin{pmatrix} I_k \\ 0 \end{pmatrix} \oplus G_nB \begin{pmatrix} G_k \\ 0 \end{pmatrix} \oplus \tilde{e}_{ii} B \begin{pmatrix} G_k \\ 0 \end{pmatrix} \\ \hline B & B \begin{pmatrix} G_k \\ 0 \end{pmatrix} \end{array} \right] \]
and we have $G_n = G_n \oplus \tilde{e}_{ii}$. Equivalently we have 
\[ G = G \oplus \tilde{e}_{i+k, i+k} \]

  \item $S$ on the input qubit $i$ ($1 \leq i \leq k$) \\
We perform the column operation $M[:,i+n] = M[:,i+n] \oplus M[:,i]$. We get 
  \[ M = \left[ \begin{array}{c|c} G_nB & (B^T)^{-1} \begin{pmatrix} I_k \\ 0 \end{pmatrix} \oplus G_nB \begin{pmatrix} G_k \\ 0 \end{pmatrix} \oplus G_nB\tilde{e}_{ii} \\ \hline B & B \begin{pmatrix} G_k \\ 0 \end{pmatrix} \oplus B\tilde{e}_{ii} \end{array} \right] \]  
  and we have $G_k = G_k \oplus \tilde{e}_{ii}$. Equivalently we have 
  \[ G = G \oplus \tilde{e}_{i, i} \]

  \item $CZ$ on the output qubits $i,j$ ($1 \leq i < j \leq n$) \\
  We perform the row operations $M[i,:] = M[i,:] \oplus M[j+n,:]$ and $M[j,:] = M[j,:] \oplus M[i+n,:]$. We get 
  \[ M = \left[ \begin{array}{c|c} G_nB \oplus \tilde{e}_{ij}B & (B^T)^{-1} \begin{pmatrix} I_k \\ 0 \end{pmatrix} \oplus G_nB \begin{pmatrix} G_k \\ 0 \end{pmatrix} \oplus \tilde{e}_{ij} B \begin{pmatrix} G_k \\ 0 \end{pmatrix} \\ \hline B & B \begin{pmatrix} G_k \\ 0 \end{pmatrix} \end{array} \right] \]
  and we have $G_n = G_n \oplus \tilde{e}_{ij}$. Equivalently we have 
  \[ G = G \oplus \tilde{e}_{i+k, j+k} \]

  \item $CZ$ on the input qubits $i,j$ ($1 \leq i < j \leq n$) \\
  We perform the column operations $M[:,i+n] = M[:,i+n] \oplus M[:,j]$ and $M[:,j+n] = M[:,j+n] \oplus M[:,i]$. We get
  \[ M = \left[ \begin{array}{c|c} G_nB & (B^T)^{-1} \begin{pmatrix} I_k \\ 0 \end{pmatrix} \oplus G_nB \begin{pmatrix} G_k \\ 0 \end{pmatrix} \oplus G_nB\tilde{e}_{ij} \\ \hline B & B \begin{pmatrix} G_k \\ 0 \end{pmatrix} \oplus B\tilde{e}_{ij} \end{array} \right] \]  
  and we have $G_k = G_k \oplus \tilde{e}_{ij}$. Equivalently we have 
  \[ G = G \oplus \tilde{e}_{ij} \]

  \item $CNOT$ on the output qubits with control $i$, target $j$ ($1 \leq i < j \leq n$) \\
  We perform the row operations $M[i,:] = M[i,:] \oplus M[j,:]$ and $M[j+n,:] = M[j+n,:] \oplus M[i+n,:]$. We get 
  \[ M = \left[ \begin{array}{c|c} E_{ij}G_nB & E_{ij} (B^T)^{-1} \begin{pmatrix} I_k \\ 0 \end{pmatrix} \oplus E_{ij} G_nB \begin{pmatrix} G_k \\ 0 \end{pmatrix} \\ \hline E_{ji} B & E_{ji} B \begin{pmatrix} G_k \\ 0 \end{pmatrix} \end{array} \right]. \]
  We write 
  \[ M = \left[ \begin{array}{c|c} E_{ij}G_nE_{ji} \times E_{ji}B &  (E_{ji}B)^{-T} \begin{pmatrix} I_k \\ 0 \end{pmatrix} \oplus E_{ij} G_nE_{ji} \times E_{ji}B \begin{pmatrix} G_k \\ 0 \end{pmatrix} \\ \hline E_{ji} B & E_{ji} B \begin{pmatrix} G_k \\ 0 \end{pmatrix} \end{array} \right] \]
  and we get $G_n = E_{ij}G_nE_{ji}$ and $B = E_{ji}B$, i.e, $B_{k,n} = B_{k,n}E_{ji}$ or $B_{k,n}^T = E_{ij}B_{k,n}^T$. Equivalently we have 
  \[ G = E_{i+k,j+k}GE_{j+k,i+k}. \]

  \item $CNOT$ on the input qubits with control $i$, target $j$ ($1 \leq i < j \leq k$) \\ 
  We perform the column operations $M[:,j] = M[:,j] \oplus M[:,i]$ and $M[:,i+n] = M[:,i+n] \oplus M[:,j+n]$. We get 
  \[ M = \left[ \begin{array}{c|c} G_nBE_{ij} & (B^T)^{-1} \begin{pmatrix} I_k \\ 0 \end{pmatrix} E_{ji} \oplus G_nB \begin{pmatrix} G_k \\ 0 \end{pmatrix} E_{ji} \\ \hline B E_{ij} & B \begin{pmatrix} G_k \\ 0 \end{pmatrix} E_{ji} \end{array} \right]. \]
  We write 
  \[ M = \left[ \begin{array}{c|c} G_nBE_{ij} & (BE_{ij})^{-T}) \begin{pmatrix} I_k \\ 0 \end{pmatrix} \oplus G_nBE_{ij} \begin{pmatrix} E_{ij}G_kE_{ji} \\ 0 \end{pmatrix}  \\ \hline B E_{ij} & BE_{ij} \begin{pmatrix} E_{ij}G_kE_{ji} \\ 0 \end{pmatrix} \end{array} \right] \]
  and we get $G_k = E_{ij}G_kE_{ji}$ and $B = BE_{ij}$, i.e, $B_{k,n} = E_{ij}B_{k,n}$. Equivalently we have 
  \[ G = E_{ij}GE_{ji}. \]
  \item if $G_n[i,i] = 0$, then we can apply an $R_x(\pi/2)$ on the output qubit $i$, we get 
  \[ M = \left[ \begin{array}{c|c} G_nB & (B^T)^{-1} \begin{pmatrix} I_k \\ 0 \end{pmatrix} \oplus G_nB \begin{pmatrix} G_k \\ 0 \end{pmatrix} \\ \hline (I + e_{ii}G_n)B & (I + e_{ii}G_n)B \begin{pmatrix} G_k \\ 0 \end{pmatrix} \oplus e_{ii}(B^T)^{-1} \begin{pmatrix} I_k \\ 0 \end{pmatrix} \end{array} \right].  \]
  One can check that $(I + e_{ii}G)^2 = I$, therefore:
  \[ \hspace*{-2cm} M = \left[ \begin{array}{c|c} G_n(I + e_{ii}G_n)(I + e_{ii}G_n)B & (((I + e_{ii}G_n)(I + e_{ii}G_n)B)^T)^{-1} \begin{pmatrix} I_k \\ 0 \end{pmatrix} \oplus G_n(I + e_{ii}G_n)(I + e_{ii}G_n)B \begin{pmatrix} G_k \\ 0 \end{pmatrix} \\ \hline (I + e_{ii}G_n)B & (I + e_{ii}G_n)B \begin{pmatrix} G_k \\ 0 \end{pmatrix} \oplus e_{ii}(((I + e_{ii}G_n)(I + e_{ii}G_n)B)^T)^{-1} \begin{pmatrix} I_k \\ 0 \end{pmatrix} \end{array} \right].  \]

  It is easy to see that we will have 
  \[ G_n \leftarrow G_n(I + e_{ii}G_n), \]
  \[ B \leftarrow (I + e_{ii}G_n)B, \]

  but it is trickier to see the effect on $G_k$. First, let's simplify the formula by replacing by the new values of $B$ and $G_n$:
  \[ M = \left[ \begin{array}{c|c} G'_nB' & (I + e_{ii}G_n)^T(B'^T)^{-1} \begin{pmatrix} I_k \\ 0 \end{pmatrix} \oplus G'_nB' \begin{pmatrix} G_k \\ 0 \end{pmatrix} \\ \hline B' & B' \begin{pmatrix} G_k \\ 0 \end{pmatrix} \oplus e_{ii}(I + e_{ii}G_n)^T(B'^T)^{-1} \begin{pmatrix} I_k \\ 0 \end{pmatrix} \end{array} \right].  \]

  We have $e_{ii}(I + e_{ii}G_n)^T = e_{ii}(I + G_ne_{ii}) = e_{ii}$ because $e_{ii}G_ne_{ii} = G_n[i,i]e_{ii} = 0$, then

  \[ M = \left[ \begin{array}{c|c} G'_nB' & (B'^T)^{-1} \begin{pmatrix} I_k \\ 0 \end{pmatrix} \oplus G_ne_{ii}(B'^T)^{-1}\begin{pmatrix} I_k \\ 0 \end{pmatrix} \oplus G'_nB' \begin{pmatrix} G_k \\ 0 \end{pmatrix} \\ \hline B' & B' \begin{pmatrix} G_k \\ 0 \end{pmatrix} \oplus e_{ii}(B'^T)^{-1} \begin{pmatrix} I_k \\ 0 \end{pmatrix} \end{array} \right].  \]  
  We remind that $G_nB = G'_nB'$, therefore $G_n = G'_nB'B^{-1}$: 
  \[ M = \left[ \begin{array}{c|c} G'_nB' & (B'^T)^{-1} \begin{pmatrix} I_k \\ 0 \end{pmatrix} \oplus G'_nB' B^{-1}e_{ii}(B'^T)^{-1}\begin{pmatrix} I_k \\ 0 \end{pmatrix} \oplus G'_nB' \begin{pmatrix} G_k \\ 0 \end{pmatrix} \\ \hline B' & B' \begin{pmatrix} G_k \\ 0 \end{pmatrix} \oplus B'B'^{-1}e_{ii}(B'^T)^{-1} \begin{pmatrix} I_k \\ 0 \end{pmatrix} \end{array} \right].  \]  
  We need to compute two last quantities: 
  \[ B^{-1}e_{ii}(B'^T)^{-1} \]
  and
  \[ B'^{-1}e_{ii}(B'^T)^{-1} \]
\ \\
  We recall that $e_{ii}(I + e_{ii}G_n)^T = e_{ii}$, and similarly $(I + e_{ii}G_n)e_{ii} = e_{ii}$. So
    \[ B^{-1}e_{ii}(B'^T)^{-1} = B^{-1} e_{ii}(I + e_{ii}G_n)^T (B^{-1})^T = B^{-1} e_{ii} (B^{-1})^T, \]
    \[ B'^{-1}e_{ii}(B'^T)^{-1} = B^{-1} (I + e_{ii}G_n)e_{ii}(I + e_{ii}G_n)^T (B^{-1})^T = B^{-1} e_{ii} (B^{-1})^T. \]

\ \\

  Writing $b_i = B^{-1}[:,i]$, $B^{-1} e_{ii} (B^{-1})^T = b_ib_i^T$ and finally we have 
    \[ M = \left[ \begin{array}{c|c} G'_nB' & (B'^T)^{-1} \begin{pmatrix} I_k \\ 0 \end{pmatrix} \oplus G'_nB' \left( b_ib_i^T \begin{pmatrix} I_k \\ 0 \end{pmatrix} \oplus \begin{pmatrix} G_k \\ 0 \end{pmatrix} \right)  \\ \hline B' & B' \left( b_ib_i^T \begin{pmatrix} I_k \\ 0 \end{pmatrix} \oplus \begin{pmatrix} G_k \\ 0 \end{pmatrix} \right) \end{array} \right].  \] 
    i.e
  \[ \begin{pmatrix} G_k \\ 0 \end{pmatrix} \leftarrow  \begin{pmatrix} G_k \\ 0 \end{pmatrix} \oplus b_ib_i^T \begin{pmatrix} I_k \\ 0 \end{pmatrix} \]
  and
  \[ G_k \leftarrow G_k \oplus b_i[1:k]b_i[1:k]^T = G_k \oplus B_{k,n}[:,i]B_{k,n}[:,i]^T. \]
  We do not care about the block below $G_k$ being modified as we can always zero it with free column operations. \\

  To summarize, we have 
  \[ G_n \leftarrow G_n(I + e_{ii}G_n), \text{ equivalently, } G_n \leftarrow G_n \oplus G_n[:,i]G_n[:,i]^T \]
  \[ G_k \leftarrow G_k \oplus B_{k,n}[:,i]B_{k,n}[:,i]^T, \]
  \[ B^{-1} \leftarrow B^{-1}(I + e_{ii}G_n), \text{ equivalently, } B_{k,n} \leftarrow B_{k,n} \oplus B_{k,n}[:,i]G_n[:,i]^T. \]

  One can check that this is equivalent to 
  \[ G \leftarrow G \oplus G[:,k+i]G[:,k+i]^T. \]

  \item A similar proof can be given to show the effect of an $R_x(\pi/2)$ gate on the input qubits.

  \item For the Hadamard gates, simply write $H \sim SR_x(\pi/2)S$ and apply successively the three gates to have the results.
\end{enumerate}

\end{document}